\documentclass[11pt]{article}
\usepackage{fullpage}
\usepackage{graphicx}
\usepackage{epstopdf}
\usepackage{imakeidx}
\usepackage{amsmath, amsthm, amssymb}
\usepackage{verbatim}
\usepackage[table]{xcolor}
\usepackage{enumerate}
\usepackage{color,graphics}
\usepackage{comment} 
\usepackage{outlines}
\usepackage{mathtools}
\usepackage{pgfplots}   
\pgfplotsset{compat=1.7}
\usepackage{pgf,tikz,tikz-qtree}   %
\usetikzlibrary{arrows, automata}  
\usetikzlibrary {positioning}
\usetikzlibrary{automata,positioning,arrows,matrix,backgrounds,calc}
\usetikzlibrary{decorations.text}
\usetikzlibrary{decorations.pathmorphing}
\usetikzlibrary{shapes.geometric}
\usepackage{algorithmicx}
\usepackage{algorithm}
\usepackage{algpseudocode}

\usepackage{multirow}
\usepackage{stackrel}
\usepackage{ifthen}
\usepackage{xparse}
\usepackage{xcolor}

\definecolor{Blue}{rgb}{0.3,0.3,0.9}
\definecolor{Black}{rgb}{0.0,0.0,0.0}

\newcommand{\ignore}[1]{}
\def \qed {\hspace*{0pt} \hfill {\quad \vrule height 1ex width 1ex depth 0pt}
 \medskip}

\newcommand{\R}{\ensuremath{\mathbb R}}

\newcommand{\E}{\ensuremath{\mathbb{E}}}
\newcommand{\dist}{\mathsf{dist}}
\newcommand{\D}[2]{\mathcal{D}_{#1}^{(#2)}}
\newcommand{\di}[2]{d_{#1}^{(#2)}}

\newcommand{\cL}{\mathcal{L}}
\newcommand{\seq}{\mathsf{seq}}
\newcommand{\den}{\mathsf{den}}

\newcommand{\cT}{\mathcal{T}}
\newtheorem{theorem}{Theorem}[section]

\newtheorem{definition}[theorem]{Definition}
\newtheorem{claim}[theorem]{Claim}

\newtheorem{observation}[theorem]{Observation}
\newtheorem{remark}[theorem]{Remark}

\newtheorem{lemma}[theorem]{Lemma}

\title{New Sublinear Algorithms and Lower Bounds for LIS Estimation}
\author{Ilan Newman\textsuperscript{*} \and Nithin Varma\thanks{E-mail: \texttt{ilan@cs.haifa.ac.il, nvarma@bu.edu}. This research was supported by The Israel Science Foundation, grant number 497/17. The second author was also supported by the PBC Fellowship for Postdoctoral Fellows by the Israeli Council of Higher Education.}}

\date{Department of Computer Science, University of Haifa, Israel.}

\begin{document}
\maketitle
\ignore{
\begin{titlepage}
\maketitle 
\def\thepage {} 
\thispagestyle{empty}

\end{titlepage}
}
\begin{abstract}
Estimating the length of the longest increasing subsequence (LIS) in an array is a problem of fundamental importance.
Despite the significance of the LIS estimation problem and the amount of attention it has received, there are important aspects of the problem that are not yet fully understood. 
There are no better lower bounds for LIS
estimation than the obvious bounds implied by testing monotonicity (for
adaptive or nonadaptive algorithms).
In this paper, we give the first nontrivial lower bound on the complexity of LIS estimation, and also provide novel algorithms that complement our lower bound.

Specifically, for every constant
$\epsilon \in (0,1)$, every nonadaptive algorithm that outputs an
estimate of the LIS length in an array of length $n$ to within an additive error of
$\epsilon n$ has to make $\log^{\Omega(\log
  (1/\epsilon))} n$ queries.  
  Next, we design nonadaptive LIS estimation algorithms whose
  complexity decreases as the number of distinct values, $r$,
  in the array decreases. 
We first present a simple algorithm that makes 
$\tilde{O}(r/\epsilon^3)$  queries and approximates the LIS length with an
additive error bounded by $\epsilon n$. 
This algorithm has better complexity than the best previously known adaptive algorithm (Saks and Seshadhri; 2017) 
for the same problem when $r \ll \text{poly}\log (n)$.
We use our algorithm to construct a nonadaptive algorithm with query complexity
$\tilde{O}(\sqrt{r}\cdot \text{poly}(1/\lambda))$ that, 
when the LIS is of
length at least $\lambda n$, outputs a multiplicative
$\Omega(\lambda)$-approximation to the LIS length. 
Our algorithm also improves upon the state of the art nonadaptive LIS estimation
algorithm (Rubinstein, Seddighin, Song, and Sun; 2019) in terms of approximation guarantee.

  Finally, we present a $O(\log n)$-query nonadaptive erasure-resilient tester  for monotonicity.
  Our result implies that lower bounds on erasure-resilient testing of monotonicity does not
  give good lower bounds for LIS estimation. 
  It also implies that nonadaptive tolerant testing is strictly harder than nonadaptive erasure-resilient testing for the natural property of monotonicity, thereby making progress towards solving an open question (Raskhodnikova, Ron-Zewi, and Varma; 2019).
\end{abstract}

\pagenumbering {arabic} 

\def \noproof
 {\hspace*{0pt} \hfill {\quad \vrule height 1ex width 1ex depth 0pt}}

\def \xor {\oplus}

\section{Introduction}

Estimating the length of the longest increasing subsequence (LIS) in an array is a problem of fundamental importance. 
For arrays of length $n$, one can solve this problem exactly in time $O(n\log n)$ using dynamic programming~\cite{Fredman75} or patience sorting~\cite{AD99}.
Approximating the length of the LIS has also been well-studied, and there are several sublinear-time algorithms~\cite{PRR06,ACCL07,SaksS10,RubinsteinSSS19} for this task.
In the approximation task, for a real-valued array $A$ of size $n$, the goal is to estimate the length 
of the LIS within an additive error (of $\epsilon n$) or multiplicative error. 
An additive $\epsilon n$-approximation algorithm for this problem can also be used to estimate, with the same approximation guarantee, the Hamming distance of $A$ to the closest sorted array\footnote{It is necessary and sufficient to modify the values that do not belong to an LIS to make the array sorted.} (a.k.a.\ distance to monotonicity).

Early sublinear-time algorithms for LIS estimation~\cite{PRR06,ACCL07} provided multiplicative $(2 + o(1))$-approximation for the distance to monotonicity, and thereby, additive $\frac{n}{2}$-approximation to the length of the LIS. 
Saks and Seshadhri~\cite{SaksS10} made a major improvement to the state of the art, and presented an algorithm that approximates the LIS length to within an additive error of $\epsilon  n$ for arbitrary $\epsilon \in (0,1)$. 
All these algorithms have query complexity polylogarithmic\footnote{The query complexity of the algorithm by Saks and Seshadhri~\cite{SaksS10} depends on the approximation parameter $\epsilon$ as $O(\left(1/\epsilon\right)^{1/\epsilon})$ and hence is within aforementioned bound only if $\epsilon$ is constant. In particular, the query complexity ceases to be sublinear as as soon as $\epsilon$ is $O(1/\log(n))$.} in $n$ for constant $\epsilon$.
Subsequently, Rubinstein, Seddighin, Song, and Sun~\cite{RubinsteinSSS19} presented a nonadaptive algorithm that computes 
a multiplicative $\Omega(\lambda^3)$-approximation to the LIS length, with query complexity $\tilde{O}(\sqrt{n} \cdot \text{poly}(1/\lambda))$, where $\lambda$ is the ratio of the LIS length to $n$.
In a very recent work (independent and parallel to ours), Mitzenmacher and Seddighin~\cite{MS21} developed a sublinear algorithm for LIS estimation with query complexity $\tilde{O}(n^{1- \Omega(\epsilon)}\cdot\text{poly}(1/\lambda))$ that obtains an approximation ratio of $\Omega(\lambda^\epsilon)$ for arbitrary $\epsilon \in (0,1)$.

Despite the significance of the LIS estimation problem and the amount of attention it has received, there are important aspects of the problem that are not yet fully understood. 
There is no better lower bound on the query complexity of LIS
estimation, for
adaptive or nonadaptive algorithms, other than the obvious bound of $\Omega(\log n)$ implied by monotonicity testing~\cite{Fischer04}.
Another issue is to investigate whether the input length $n$ is the right parameter to express the complexity of LIS estimation algorithms. In other words, it is unknown whether there are other input parameters that capture the fine-grained complexity of LIS estimation by making use of the underlying combinatorics of the problem.

In this paper we address both these issues.
We prove the first nontrivial lower bound on the query complexity of nonadaptive algorithms for additive error LIS estimation. We also design nonadaptive LIS estimation algorithms whose query complexity is parameterized in terms of the number of distinct values in the input array.

\smallskip

{\it \textbf{Lower Bound for LIS Estimation:}} We show that there is no nonadaptive algorithm that approximates the LIS length to arbitrary additive error and has query complexity polylogarithmic in $n$.
Specifically, for arbitrary constant $\epsilon \in (0,1)$, every nonadaptive LIS estimation algorithm that has an additive error bounded by $\epsilon n$ has to make $\log^{\Omega(\log
  (1/\epsilon))} n$ queries. 
Interestingly, our lower bound construction uses ideas from the lower
bound~\cite{BCLW19} on the query complexity of 1-sided error nonadaptive testers for
the property of $(k,\dots,2,1)$-freeness.
 This is the first lower
 bound that improves upon the obvious lower bound of $\Omega(\log n)$. 

One general approach for proving lower bounds on the complexity of LIS estimation
was proposed by Dixit, Raskhodnikova, Thakurta, and Varma~\cite{DixitRTV18},  who showed that lower bounds for erasure-resilient testing of monotonicity
 provides lower bounds for estimating the distance to monotonicity up to an additive error. 
We prove that this method cannot provide a nontrivial lower bound for LIS
estimation, by showing a $O(\log n)$-query \emph{nonadaptive} algorithm for
erasure-resilient monotonicity testing.

\smallskip

{\it \textbf{Sublinear Algorithms for LIS Estimation:}}
Our starting point here is to understand the dependence of the query complexity of LIS estimation 
on the range size of an input array.
This is a major direction of study for the simpler
problem of monotonicity testing,
since the only tight lower bound~\cite{Fischer04} holds for exponential range.
Recently, Pallavoor, Raskhodnikova, and Varma~\cite{PRV18}, and Belovs~\cite{Belovs18}, gave efficient algorithms
for monotonicity testing whose query complexity beats the above lower bound when range size is small.
There were no explicit results on LIS estimation for
limited range size before our work.\footnote{For the case of Boolean
  arrays, Berman, Raskhodnikova, and Yaroslavtsev~\cite{BRY14} showed that one can approximate
  the LIS length to within an additive error of $\epsilon n$ by
  making $O(1/\epsilon^2)$ queries.}
  In this paper, we give efficient \emph{nonadaptive} LIS estimation algorithms whose complexity is parameterized by 
  $r$, the number of distinct values in the array, which is always at most the range size. 
  Our algorithms improve upon the state of the art algorithms in both complexity 
  and approximation guarantee when the range is small.

 We first show a $\tilde{O}(r/\epsilon^3)$-query nonadaptive 
algorithm  for LIS estimation, of additive error $\epsilon
n$,  for arbitrarily small $\epsilon$.
In particular, when the LIS length is a constant fraction of $n$, our algorithm can be used to get a multiplicative $(1 \pm \epsilon)$-approximation for the LIS length.
We add that our algorithm is the only sublinear nonadaptive algorithm giving this approximation guarantee when $r = o(n)$. 
Furthermore, when $r = o(\log^k n)$ (for an appropriate power $k$),
our algorithm outperforms the adaptive algorithm of Saks and Seshadhri~\cite{SaksS10}, not only in terms of the dependence of query complexity on the input size $n$, but also in terms of its dependence on the approximation parameter $\epsilon$.
Hence, our algorithm bridges the gap between the known 
$\Omega(\text{poly} \log n)$-query algorithm for the general range and the
$O(1)$-query algorithm for the Boolean range.

An additional main result of this paper is a $\tilde{O}(\sqrt{r})$-query nonadaptive algorithm that gives a multiplicative approximation to the LIS length even when the LIS is relatively small. Namely, the
algorithm makes $\tilde{O}(\sqrt{r}\cdot\text{poly}(1/\lambda))$ queries and 
outputs a multiplicative $\Omega(\lambda)$-aproximation to the LIS length,
where $\lambda$ denotes the LIS length normalized by the input length. 
This is an improvement over the 
algorithm by Rubinstein, Seddighin, Song, and Sun~\cite{RubinsteinSSS19}, which makes $\tilde{O}(\sqrt{n}\cdot\text{poly}(1/\lambda))$ nonadaptive queries and outputs a multiplicative $\Omega(\lambda^3)$-approximation to the LIS length.
 Our algorithm improves upon \cite{RubinsteinSSS19} 
in terms of approximation guarantee (even in the general case of $r = n$) as well as query complexity (when $r \ll n$), and further, works
for any value of $r$.  
Finally, the query complexity of our algorithm is always better than that of the recent LIS estimation algorithm by Mitzenmacher and Seddighin~\cite{MS21} that outputs a multiplicative $\Omega(\lambda^\epsilon)$-approximation to the LIS length for arbitrary $\epsilon \in (0,1)$.\footnote{We point out that the LIS estimation algorithm of Mitzenmacher and Seddighin~\cite{MS21} uses the algorithm of Rubinstein et al.~\cite{RubinsteinSSS19} as a subroutine. By using our algorithm instead, the query complexity of the algorithm of Mitzenmacher and Seddighin~\cite{MS21} can be improved.} 
\smallskip

\textit{\textbf{Separating Distance Estimation from Erasure-Resilient Testing:}}
As mentioned before, a general method for proving lower bounds on
distance estimation (or tolerant testing~\cite{PRR06}) is via proving lower bounds on erasure-resilient testing~\cite{DixitRTV18}.

Our nonadaptive erasure-resilient tester for monotonicity with complexity $O(\log n)$ and 
our lower bound on the query complexity of nonadaptive algorithms for LIS estimation imply that nonadaptive tolerant testing is strictly harder than nonadaptive erasure-resilient testing for the natural property of monotonicity, thereby making progress towards solving an open question raised by Raskhodnikova, Ron-Zewi, and Varma~\cite{RRV19}.
\subsection{Discussion of Results and Overview of Techniques}

In this section, we state our results more formally,
 and provide an overview of the
techniques used to prove them. 
We use ideas
from~\cite{RubinsteinSSS19}, \cite{NewmanRRS19} and \cite{BCLW19}.

Given a real-valued array $A$ of length $n$, an LIS in $A$ is the longest nondecreasing sequence of values in $A$. 
In other words, the LIS is a largest cardinality set $\mathcal{L}$ of
indices such that for $u, v \in \mathcal{L}$, we have $u < v$ if and
only if $A[u] \le A[v]$. We abuse notation and also use the term LIS
to denote $|\mathcal{L}|$ when this is clear from the context.
A real-valued array of length $n$ can be equivalently viewed
as a function from $[n]$ to the reals. 
Adopting this view, we use the term \emph{monotone array} to refer to a
sorted array.
Throughout, we denote by $r$, the number of (or a
guaranteed upper bound on) distinct values in the array. That is,
$r=|R|$ for   
 $R = \{A[i]:~ i \in [n]\}.$ Thus, for the
unrestricted case it is assumed that $r=n$.  

\subsubsection{Lower bound on the query complexity of nonadaptive LIS
  estimation algorithms}

  Our first result proves that there is no nonadaptive algorithm that approximates the length of the LIS in an array of length $n$ to within an additive error of $\epsilon  n$, for arbitrary constant $\epsilon \in (0,1)$, and has query complexity polylogarithmic in $n$.
\begin{theorem}\label{thm:adaptive-tolerant-lb}
For every $\epsilon \in (0,1)$, every nonadaptive algorithm that on an array $A$ of length $n$, outputs an additive $\epsilon n$-approximation to the length of the LIS in $A$, has to make $\log^{\Omega(\log(1/\epsilon))} n$ queries.
\end{theorem}

We note that this is the first lower bound on LIS estimation that is
not directly implied by the lower bound for testing
monotonicity~\cite{Fischer04}.

To prove our lower bound, we construct two distributions with different
LIS lengths such that every deterministic nonadaptive
algorithm distinguishing the distributions with probability at least
$2/3$, has query complexity $\log^{\omega(1)} (n)$. 
More specifically, for every natural number $h$, we construct
distributions $\mathcal{D}_0^{(h)}$ and $\mathcal{D}_1^{(h)}$ that are
supported on inputs whose LIS lengths differ by $\exp(-h)$. We then
prove  that every deterministic nonadaptive algorithm that
takes input from the union of the supports of  $\mathcal{D}_0^{(h)}$ and
$\mathcal{D}_1^{(h)},$ and aims to correctly identify the distribution from which
the input is taken, either fails for most inputs or makes $\Omega(\log^h n)$ queries. 
Interestingly, our lower bound construction uses ideas from the lower
bound of Ben-Eliezer, Canonne, Letzter, and Waingarten~\cite{BCLW19} on the query complexity of $1$-sided error nonadaptive testers for
the property of $(k,\dots,2,1)$-freeness, where an array
$A$ of length $n$ is $(k,\dots,2,1)$-free if there are no $k$ indices
$i_1 < i_2 < \dots < i_k$ such that $A[i_1] > A[i_2] > \dots >
A[i_k]$.

\paragraph*{Using reductions from erasure-resilient testing:} 
As mentioned before, a general method for proving lower bounds on
distance estimation is via proving lower bounds on erasure-resilient testing~\cite{DixitRTV18}.

\begin{definition}[Erasure-resilient monotonicity tester]~\label{def:er-testing}
Given
$\epsilon, \alpha \in (0,1)$
and
a real-valued array $A$ containing at most
$\alpha$-fraction of \emph{erased} values\footnote{Erasures are made
  adversarially before the tester makes its queries and the tester is
  unaware of the location of the erasures. A tester that queries the
  value at an erased location is returned a special symbol $\perp$.},
the goal of an $\alpha$-erasure-resilient $\epsilon$-tester for
monotonicity is to determine 
whether
$A$ can be completed to a monotone array or whether every completion of
$A$ has Hamming distance at least $\epsilon n$ to monotonicity.
\end{definition}
 
 Dixit, Raskhodnikova, Thakurta and Varma~\cite{DixitRTV18} observed that the
complexity of erasure-resilient (ER) testing  a property, falls in
between the complexity of standard testing the property and
estimating the distance to that property (with additive error). 
Hence, a lower bound on the complexity of
ER testing monotonicity implies the same lower bound
for estimating the length of the LIS up to an additive error.  The
only previously known ER tester for monotonicity ~\cite{DixitRTV18} is adaptive and has query
complexity $O(\log(n)/\epsilon)$.
Hence, a nontrivial lower bound for (adaptive) LIS estimation cannot be obtained
this way. 

We present  a {\em nonadaptive} ER tester that makes $O(\log
n)$ queries and works for all fraction of erasures. This makes the
results on ER testing monotonicity tight, and also shows
that one cannot obtain a lower bound for LIS estimation via ER testing
even for nonadaptive algorithms.

\begin{theorem}\label{thm:nonadaptive-erasure-resilient-monotonicity-test}
Let $\epsilon, \alpha \in (0,1)$ such that $\alpha + \epsilon <
1$. There exists a nonadaptive $\alpha$-erasure-resilient
$\epsilon$-tester for monotonicity  with query complexity
$O\left(\frac{\log n}{\epsilon^2} + \frac{1}{\epsilon^3}\right)$, for
arrays for length $n$. 
\end{theorem}

The ER testers designed by Dixit et
al.~\cite{DixitRTV18} for various properties, are all either adaptive, or
obtained by repeating a (standard) tester that makes independent and uniformly
distributed queries.  
Our tester is different, and is in this sense,
the first nontrivial nonadaptive ER tester for a
natural property. 
Consider an array $A$ of length $n$ with at most $\alpha$ fraction of
erasures, where $\alpha \in [0,1)$.  Our tester samples an index
$s \in [n]$ uniformly at random and does a randomized binary search for
$s$ on the array as if it were monotone.  It queries the array values on
these indices, and looks for violations to monotonicity on the search path to
$s$. In case there are no erasures, this is a good
strategy to detect a violation to monotonicity~\cite{EKK+00}. However, when values at a constant fraction
of indices are erased, it could be the case that most of the values on
the search path are erased. We show that a slightly modified version
of this tester can be used for testing monotonicity.  Specifically, our
tester, in addition to querying the values along the binary search
path, also queries the indices in a small constant-sized interval
around the search point $s$.  To analyze this modified tester, we rely
on a combinatorial lemma by Newman, Rabinovich, Rajendraprasad, and Sohler~\cite{NewmanRRS19}.  A
nonerased index $x \in [n]$ is $\gamma$-deserted for
$\gamma \in (0,1)$ if there exists an interval $I \subseteq [n]$ such
that $x \in I$ and at most $\gamma$ fraction of the values in $I$ are
nonerased.  Roughly speaking, the lemma implies that the fraction of
$\gamma$-deserted indices in $A$ is proportional to
$\gamma\cdot\alpha$.  Using this, we are able to argue that, with high
probability, the index $s$ that we sample as the search point is not
$\gamma$-deserted (for an appropriate choice of $\gamma$) and that it
forms a violation with enough other nonerased indices, so as to ensure a
high probability of success.

\subsubsection{Parameterized and nonadaptive algorithms for LIS estimation}

We present efficient nonadaptive LIS estimation algorithms.
The novelty is that we parameterize the complexity of LIS estimation algorithms in terms of the number of distinct values $r$ in an array.

\smallskip
\noindent We first show an LIS estimation algorithm with query complexity $\tilde{O}(r)$. 
\begin{theorem}\label{thm:O(r)}
There exists a nonadaptive algorithm that, given a real-valued array
$A$ of length $n$ containing at most $r$ distinct values, and a
parameter $\epsilon \in (0,1)$, makes $\tilde{O}(r/\epsilon^3)$
queries and outputs, with probability at least $2/3$, an estimate for the LIS size that is accurate to
within additive $\epsilon n$-error.    
Moreover, the queries of the algorithm are uniformly and independently distributed.
\end{theorem}

We mention that the approximation guarantee provided by the algorithm
is quite strong and holds even for non-constant error parameter
$\epsilon$. It matches the approximation guarantee of the
adaptive LIS estimation algorithm by Saks and
Seshadhri~\cite{SaksS10}, which makes $\text{polylog}(n)$ queries when
$\epsilon = \theta(1)$.
In particular, when the length of the LIS $\mathcal{L}$ is a constant fraction of $n$, our algorithm can be used to get a multiplicative $(1 \pm \epsilon)$-approximation for the LIS length.
We add that our algorithm is the only nonadaptive sublinear-time algorithm giving this approximation guarantee as soon as $r = o(n)$. 
Furthermore, when $r = o(\log^k n)$ (for an appropriate power $k$), our algorithm performs much better than the algorithm of Saks and Seshadhri, not only in terms of the dependence of query complexity on the input size $n$, but also in terms of the dependence on the approximation parameter $\epsilon$.

The  high level idea of the algorithm is that it is enough to restrict
attention to special subarrays that are \emph{dense} and
\emph{nice}, as elaborated in the following. 
Let $\mathcal{L}$ be a fixed unknown LIS in the input array.
A subarray is dense if a constant fraction of its indices
belong to $\mathcal{L}$, and 
it is nice if the LIS takes at most one distinct value in the subarray.   
Informally, we divide the array into $O(r/\epsilon)$ subarrays.
This will make most dense subarrays nice with respect to $\mathcal{L}$ (for an appropriate density parameter).
We then sample $O(\log r)$ indices in each
subarray to find the values that are `typical' in each subarray.

Next, our goal is to output as an estimate for $|\mathcal{L}|$, the size of $\mathcal{L}'$, which is the restriction of
$\mathcal{L}$ to such typical values. This will naturally be an
underestimate, but with a small additive error. To estimate the
size of $\mathcal{L}'$, we consider all possible increasing sequences
of the typical values, taking one value from each subarray. Since most
subarrays are nice, the size of $\mathcal{L}'$ restricted to such a sequence of
values is quite close to $|\mathcal{L}'|$. Finally, for a given
nice subarray $A_i$, the largest subsequence in $A_i$ that takes
one given value $v$ can be easily determined -- this is just the distance
to the array taking the value $v$ everywhere.

\medskip
\noindent Next, we use the above $\tilde{O}(r)$-query algorithm to obtain a nonadaptive LIS estimation algorithm with query complexity $\tilde{O}(\sqrt{r})$. 

\begin{theorem}\label{thm:O(root-r)}
There exists a nonadaptive algorithm that, given a real-valued array
$A$ of length $n$ containing at most $r$ distinct values and
$|\emph{LIS}(A)| = \lambda \cdot n$, makes $\tilde{O}(\sqrt{r}\cdot
\emph{poly}(1/\lambda))$ queries and outputs, with probability at
least $2/3$, an estimate $\mathsf{est}$ such that $\Omega (\lambda\cdot |\emph{LIS}(A)|) \le \mathsf{est} \le O(|\emph{LIS}(A)|).$
\end{theorem}

As mentioned before, this  result is an improvement over a recent LIS
estimation algorithm by Rubinstein, Seddighin, Song and Sun~\cite{RubinsteinSSS19}, 
in terms of the approximation guarantee. 
Additionally, the complexity of our algorithm improves as the number of distinct values in the input array decreases. 
Another advantage of our algorithm (also that of~\cite{RubinsteinSSS19}) is that its query complexity is sublinear, even if $\lambda$ is sub-constant.  

Our $\tilde{O}(\sqrt{r})$-query nonadaptive algorithm is somewhat
complicated. In the following, we present a high-level description of the algorithm.
We denote the input array by $A$ and use $\cL$ to 
denote a fixed LIS in $A$. 
We visualize the array values as points in an
$r \times n$ grid $G_n$. 
The vertical axis of $G_n$ represents the range $R$ of the array and is labeled with the at most $r$
distinct array values in increasing order and the horizontal axis is labeled with the
indices in $[n]$. 
We refer to an index-value pair in the grid as a point.
The grid has $n$ points, to which we do not have direct access.
We use queries to the array to form some \emph{approximate}
picture of the location of points in this grid, and use it
to estimate $|\cL|$.

The main idea is to build, in $\tilde{O}(\sqrt{r})$ queries, a data
structure that possesses enough information to compute an estimate $\mathsf{est}$,
which is a lower bound
on $|\mathcal{L}|$ and is also a reasonably good approximation. 
Roughly speaking, the first step in building this data structure is the following. 
We divide the $r\times n$ grid $G_n$ into $y^*$ rows and $x$ columns that 
partitions $G_n$ into a $y^*  \times x$ grid $G'$ of boxes, where $y^* =
\Theta (\sqrt{r})$ and $x = \Theta (\sqrt{r})$ (both depend also 
on $\lambda$).
Specifically, we divide the interval $[n]$ into $x$ contiguous subarrays.
For $i \in [x]$, let $D_i$ denote the $i$-th subarray.
Additionally, we divide the range $R$ into $y^*$ contiguous intervals of array values,
where for $j \in [y^*]$, we use $I_j$ to denote the $j$-th interval when the intervals are 
sorted in the nondecreasing order of values.
The set of boxes in $G'$ is then $\{(I_j,D_i): ~
i \in [x], j \in [y^*]\}$.

For simplicity, we assume that $r=n$ for the rest of the high-level
description.
The $y^* \times x$ grid of boxes $G'$ induces a poset on the $y^* x$
boxes, which is similar to the natural poset defined on $G_n$. 
Namely, for two boxes
in $G'$ (or for two points in $G_n$), we have $(I_j,D_i)
\preceq (I_t, D_s)$ (or $(i,j)
\leq (t,s)$) if $i \leq s$ and $j \leq t$.
The points in $\mathcal{L}$ form a chain in the
above poset in $G_n$.  Conversely, each chain in the poset
$G_n$ forms an increasing subsequence in the array $A$. Further, the boxes in
$G'$ through which $\mathcal{L}$ passes also forms a chain in the poset in
$G'$. 
On the other hand, every chain of boxes in the poset in $G'$ induces a number of chains
in the poset in $G_n$, but of possibly quite different lengths. Our strategy is
to find a small collection of chains in the poset in $G'$ that cover all boxes through which
the fixed $\mathcal{L}$ passes, and then to estimate the length of an LIS in
each of these chains of boxes. 

Let $I \subseteq R$ be a subset of the range $R$ of values and $B$ be a subarray of $A$. The density of the box $(I,B)$, denoted by $\den(I,B)$, is defined to be
the fraction of indices in the subarray $B$ whose values belong to the interval $I$. 
In other words, for each box $(I_j,D_i) \in
G',$ its density $\den(I_j,D_i)$ is the
fraction of indices in the subarray $D_i$ whose values land in the interval $I_j$.
For $\beta < 1$, a box
$(I_j, D_i)$ is said to be $\beta$-dense, if $\den(I_j,D_i) \geq \beta$.
There can be at most $\frac{1}{\beta}$ boxes that are $\beta$-dense in
any particular subarray $D_i$.

Suppose that we know (a good approximation of) the density of every
box in $G'$ (this is what we require from our data structure, and this
will be achieved via sampling). 
Then,  
we may restrict our attention to the at most
$x/\beta$ dense boxes in $G'$ and compute the LIS only in the
corresponding part of $G_n$. This is obviously  an underestimate
of the size of $\mathcal{L}$,  but one that can be afforded;  deleting
every box that is not $\beta$-dense from the chain of boxes that
$\mathcal{L}$ passes through causes the  deletion at most $\beta n$ points from
$\mathcal{L}$.

We note that the same global idea is also used in the algorithms of \cite{RubinsteinSSS19}
and (implicitly) also of \cite{SaksS10}, but in a completely different
setting (and grid sizes) which makes the first one weaker in term of
approximation guarantees, and the second one necessarily adaptive.

Next, in order to further reduce the number of possible chains of boxes in
which we need to compute LIS, we note that we can delete large
antichains of boxes from $G'$, while not decreasing the LIS size by much.
For this, we first consider a finer partition of each dense box into dense cells of
nearly equal densities, and then define a poset on the set of all
dense cells in the whole array.
We then remove large antichains from this latter poset 
and argue that the removal of dense cells 
participating in these antichains does not
\emph{hurt} the LIS too much.
Finally, by using Dilworth's theorem, we are able to obtain a collection of 
 a {\em constant} number of chains in $G',$ that covers
the restriction of $\mathcal{L}$ to the undeleted boxes.

The next idea is to estimate the LIS in each of the constantly many 
remaining chains. This results in a loss of a multiplicative
constant factor in the LIS size estimation, which we can afford.

At this point, we have reduced the problem to the estimation
of the LIS in a given fixed chain of $\beta$-dense boxes in $G'$.
Such a chain can be partitioned into two chains, one that contains only
strictly horizontal chains on disjoint subarrays, and the other that contains
only strictly vertical chains on disjoint interval ranges (see Figure
\ref{fig:chain1}). We will estimate the LIS in each, losing possibly
another multiplicative $2$-factor, which we are prepared to accept.

The final idea is the following. For the vertical going chain, one can
just sample a constant number of vertically going subchains, estimate the LIS length in each
one of them, and use these
estimates to estimate the LIS length in the vertical chain. 
By the Hoeffding bound, this
will be a good approximation. 
When $r = n$, estimating the LIS in a single vertical going subchain is trivial; we just query all
  $n/x = \tilde{O}(\sqrt{n})$ points in the subarray $D_i$
  corresponding to that subchain. 
  For smaller $r$, this 
is  not possible, and what we do is to reduce to the 
algorithm implied in Theorem~\ref{thm:O(r)}, using the fact that most
vertically going chains span a small range (this later fact will 
have to be argued from the way the data structure is formed).

For horizontally going chains, we will need a bit more from our data
structure.  The partition $G'$ of $G_n$ will
  be such that every layer formed by $i \in [y^*]$ contains either a small
  fraction of points from $\mathcal{L}$, or it contains only one range
  value. This is the only place in which we actually make use of the
fact that $y^* = \Omega(\sqrt{n})$, which lower bounds the query
complexity of the algorithm.  Having this guarantee on the grid $G'$, 
   it would have been enough
  to sample a constant number of horizontally going subchains, and
  estimate the LIS within. Further,  by the guarantee above,  each horizontal layer in $G'$
  contains only a small number of distinct values. This implies that  we could again employ our algorithm
  of Theorem~\ref{thm:O(r)}. However, this will make the whole
  algorithm adaptive (as one has to `locate' the horizontal
  segments). Instead, we show that we can concentrate on short
  (spanning a constant number of boxes) subchains, which will allow us
  to employ the algorithm given by Theorem~\ref{thm:O(r)} nonadaptively.

\subsubsection{Separating erasure-resilient testing from tolerant testing}

Tolerant testing is a generalization of property testing~\cite{RubinfeldSudan96,GGR98} defined by Parnas, Ron and Rubinfeld~\cite{PRR06}.
Specifically, a $(\delta, \epsilon + \delta)$-tolerant tester  of monotonicity distinguishes, with probability at least $2/3$, between the cases that the distance of $A$ to monotonicity is less than $\delta  n$ and at least $(\epsilon + \delta)  n$, where $\epsilon, \delta \in (0,1)$.

It has been observed~\cite{PRR06} that a tolerant tester for a property is equivalent to an algorithm for
estimating the distance to that property with an additive error guarantee.
Hence, the task of estimating the LIS up to an additive error is equivalent to tolerant monotonicity testing.
This allows us to restate Theorem~\ref{thm:adaptive-tolerant-lb} in terms of tolerant testing as follows.
\begin{theorem}\label{thm:estimation-to-tol}
For every $\epsilon \in (0,1)$, there exists a constant $\delta \in
(0,1)$ such that every nonadaptive $2$-sided error $(\delta, \delta +
\epsilon)$-tolerant  tester of monotonicity has query complexity $\log^{\Omega(\log(1/\epsilon))} n$.
\end{theorem}
Theorem~\ref{thm:estimation-to-tol} and
Theorem~\ref{thm:nonadaptive-erasure-resilient-monotonicity-test}
together imply that for the property of monotonicity, nonadaptive
tolerant testing is strictly harder than nonadaptive ER testing, and also significantly less efficient than adaptive tolerant testing.
Our results make progress towards answering the open question raised by Raskhodnikova, Ron-Zewi, and Varma~\cite{RRV19} on the existence of natural properties for which one can show a separation between tolerant testing and ER testing in terms of query complexity.


\subsection{Organization}
We describe our notation in Section~\ref{sec:prelims}. 
Our lower bound on the query complexity of nonadaptive algorithms for LIS estimation is
presented in Section~\ref{sec:adaptive-lb}. 
In Section~\ref{sec:na-er-monotonicity-tester}, we describe our nonadaptive 
erasure-resilient monotonicity tester.
Our nonadaptive and parameterized algorithms for LIS estimation are presented in Section~\ref{sec:nonadaptive-LIS-algorithms}.

\section{Notations and Preliminaries}\label{sec:prelims}

For a natural number $n$, we use $[n]$ to denote the set $\{1, 2, \dots, n\}$.
For a real-valued array $A$ of length $n$, we use $A[i]$ to denote the $i$-th entry of $A$ for $i \in [n]$.
For $x \leq y \in [n]$ we denote by  $[x,y]$ the set $\{x, x+1, \dots,
y\}$. 
The array $A$ is monotone if for every two indices $u, v \in [n]$ such that $u < v$, we have $A[u] \le A[v]$.
If $A$ is not monotone, two indices $u,v \in [n]$ are said to violate monotonicity if $u < v$ and $A[u] > A[v]$.
For $\epsilon \in (0,1)$, we say that $A$ is $\epsilon$-far from monotone if the values on at least $\epsilon \cdot n$ indices need to be modified to get{} a monotone array. 
$A$ is $\epsilon$-close to monotone if there is a way to modify the values on fewer than $\epsilon \cdot n$ indices to get a monotone array.
For a parameter $\alpha \in [0,1)$, we say that $A$ is $\alpha$-erased, if at most $\alpha$ fraction of the array values evaluate to a special symbol $\perp$.
An assignment of values to the erased points in an array is called a completion.
An $\alpha$-erased array is monotone if there exists a completion that is monotone; it is $\epsilon$-far from monotone if every completion is $\epsilon$-far from monotone.
We assume that algorithms accesses an input array $A$ via an oracle; that is when the algorithm makes a query $i \in [n]$, the oracle returns a special symbol $\perp$ if the array value at index $i$ is erased, and $A[i]$ otherwise.
An algorithm is \emph{adaptive} if its queries depend on the answers to its previous queries, and is \emph{nonadaptive} otherwise.

A partially ordered set (poset) is a set $\mathcal{P}$ associated with a reflexive, transitive, antisymmetric order relation $\preceq$ on its elements.
We denote the poset by $\langle\mathcal{P}, \preceq \rangle$. 
A chain in $\langle\mathcal{P}, \preceq \rangle$ of length $k$ is a sequence of elements $x_1 \preceq x_2 \preceq \dots \preceq x_k$.
An antichain is a set $S \subseteq \mathcal{P}$ such that for $u,v \in S$ neither $u \preceq v$ nor $v \preceq u$.

\section{Lower Bounds for Nonadaptive LIS Estimation}\label{sec:adaptive-lb}

In this section, we prove our lower bounds (Theorem~\ref{thm:adaptive-tolerant-lb}) on the query complexity
of $2$-sided error nonadaptive algorithms for estimating the distance
of real-valued arrays of length $n$ from monotonicity up to an additive error of $\epsilon \cdot n$ for some constant $\epsilon \in (0,1)$.
Equivalently, our lower bounds also hold for algorithms that $(\delta, \epsilon + \delta)$-tolerant test monotonicity for constants $\delta, \epsilon \in (0,1)$.
Interestingly, our lower bounds use ideas from the lower bound on the query complexity of 1-sided error nonadaptive testers for the property of $(k,\dots,2,1)$-freeness~\cite{BCLW19}, where an array $A$ of length $n$ is $(k,\dots,2,1)$-free if there are no $k$ indices $i_1 < i_2 < \dots < i_k$ such that $A[i_1] > A[i_2] > \dots > A[i_k]$.  

An algorithm is said to be \emph{comparison-based} if its decisions are based only on the ordering relation between the queried values, and not on the values themselves. 
The following Lemma~\ref{lem:fischer}, which follows from the work of Fischer~\cite{Fischer04}, states that it is enough to restrict our attention to comparison-based algorithms.
\begin{lemma}\cite{Fischer04}\label{lem:fischer} 
There is an optimal comparison-based algorithm for computing an additive $\epsilon n$-approximation to the LIS in real-valued arrays of length $n$ for all constant $\epsilon \in (0,1)$.
\end{lemma}
Even though Fischer~\cite{Fischer04} proves the above statement in the context of testing monotonicity in the standard model, his proof also works for the case of tolerant testing monotonicity, and in turn for LIS estimation. 
In the rest of this section, we restrict our attention to comparison-based algorithms for LIS estimation.

\subsection{An $\Omega(\log^2 n)$ Lower Bound}\label{sec:log-squared-lower-bound}

As a starting point, we prove an $\Omega(\log^2 n)$ lower bound. 
Throughout this section, we assume that $n$ is of the form $2^{2^x}$ for some natural number $x$.

We prove the lower bound using Yao's method.
Specifically, we describe two distributions $\mathcal{D}_0$ and $\mathcal{D}_1$ over real-valued arrays of length $n$, with different distances to monotonicity, and show that every deterministic nonadaptive comparison-based algorithm distinguishing these distributions with probability at least $2/3$, has to make $\Omega(\log^2 n)$ queries. 

For ease in describing our distributions, we first define some notation. 
We think of the indices of an array of length $n$ as the leaves of an ordered binary tree $\mathcal{T}$ of height $\log(n) + 1$.
We associate bit positions in the $\log n$-bit representation of the numbers in $[n]$ with the non-leaf nodes of $\mathcal{T}$.
The root is associated with the most significant bit (or the bit position $\log n$).
Every node at distance $i \in [\log (n) - 1]$ from the root is associated with bit position $\log (n) - i$.
The bit position associated with a node is also referred to as its
\emph{level} (and the level of the root is $\log n$).
The edges connecting a node with its left and right children are labeled $0$ and $1$, respectively.
Clearly, the string obtained by concatenating all the edge labels on a root-to-leaf path in $\mathcal{T}$ gives the binary representation of the index corresponding to the leaf.
For two indices $x,y \in [n]$ (which are, by definition, the leaves in $\mathcal{T}$), we use $\mathsf{LCA}(x,y)$ to denote the lowest common ancestor of $x$ and $y$ in $\mathcal{T}$.
For $j \in [\log n]$  there are obviously $n/2^j$
subtrees of $\mathcal{T},$ each rooted at level $j$. When ordered
from left to right, these $n/2^j$ subtrees partition $[n]$ into blocks
of size $2^j$. The ``the $\ell$-th
$j$-block'' refers to the $\ell$-th block from left whose size is $2^j$. 

\paragraph{The distributions $\mathcal{D}_0$ and $\mathcal{D}_1$.} We first describe the steps that are common to constructing the distributions $\mathcal{D}_0$ and $\mathcal{D}_1$. 
Sample numbers $j_1, j_2 \in [\log n]$ such that $j_1 < \log (n) - 14,$ and $j_2 < j_1 - 14$. 
\ignore{
$j_1 < \log(n) - \log (h+2) - (4h+4)$, and $j_2 < j_1 - \log (h+2) - (4h+4)$.
}
We refer to the numbers $j_1, j_2$ as the \emph{scales} of the distributions.

We start with the monotone array $A$ in which $A[u] = u$ for all $u \in [n]$. 
Swap the array values between the left and right halves of every
$j_1$-block. See the left part of Figure~\ref{fig:D0block} to see how the relative values in each $j_1$-block
of the array will look like at this point.  
For $\ell \in [n/2^{j_1}]$, let $B_\ell$ denote the $\ell$-th $j_1$-block. 
\begin{itemize}
  \item {\bf Distribution} $\mathcal{D}_0$: 
  Independently for each $\ell \in [n/2^{j_1}]$: 
  \begin{enumerate}
    \item with probability $\frac{1}{2}$, for each $j_2$-block inside $B_\ell$, swap the array values between the left and right halves of that $j_2$-block. 
  \end{enumerate}
  \item {\bf Distribution} $\mathcal{D}_1$: 
  Independently for each $\ell \in [n/2^{j_1}]$:
  \begin{enumerate}
  \item with probability $\frac{1}{2}$, for each $j_2$-block inside the left half of $B_\ell$, swap the array values between the left and right halves of that $j_2$-block,
  \item with the remaining probability $\frac{1}{2}$, for each $j_2$-block inside the right half of $B_\ell$, swap the array values between the left and right halves of that $j_2$-block.
  \end{enumerate} 
\end{itemize}

The relative values taken by the array $A$ on indices in an arbitrary $j_1$-block 
in distributions $\mathcal{D}_0$ and $\mathcal{D}_1$ can be
visualized as in Figures~\ref{fig:D0block} and~\ref{fig:D1block}. 
We note that in both $\mathcal{D}_0,\mathcal{D}_1,$ all values in the $\ell$-th $j_1$-block are
smaller than the values in the $(\ell+1)$-th $j_1$-block for all $\ell \in [(n/2^{j_1}) - 1]$.

  \begin{figure}
    \begin{center}
    \includegraphics[scale=0.3]{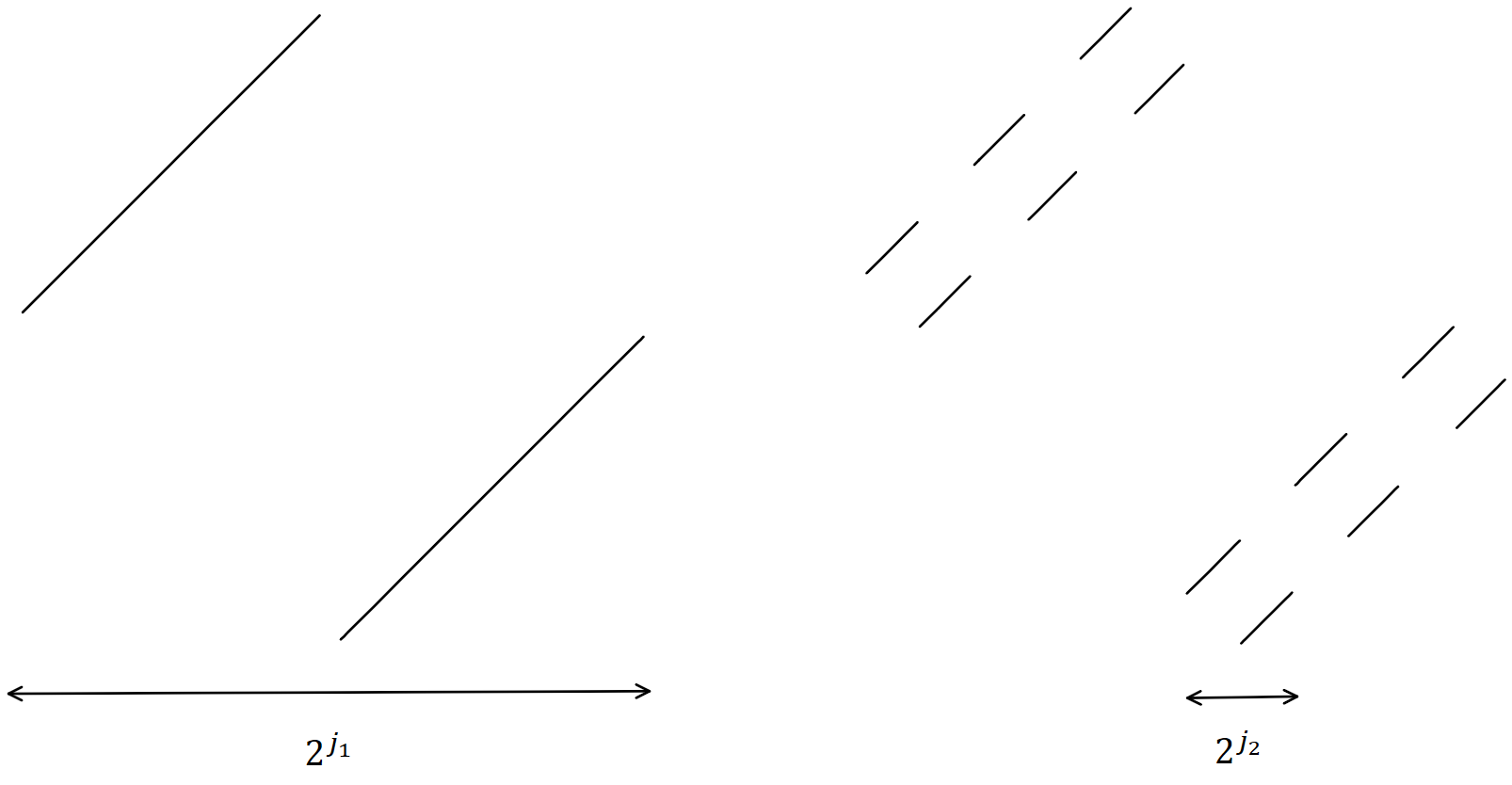}
  \end{center}
    \caption{The relative values in $j_1$-blocks of $\mathcal{D}_0$ either look like the left or right diagrams above, with equal probability.}
    \label{fig:D0block}
    \end{figure} 
  \begin{figure}
  \begin{center}
    \includegraphics[scale=0.3]{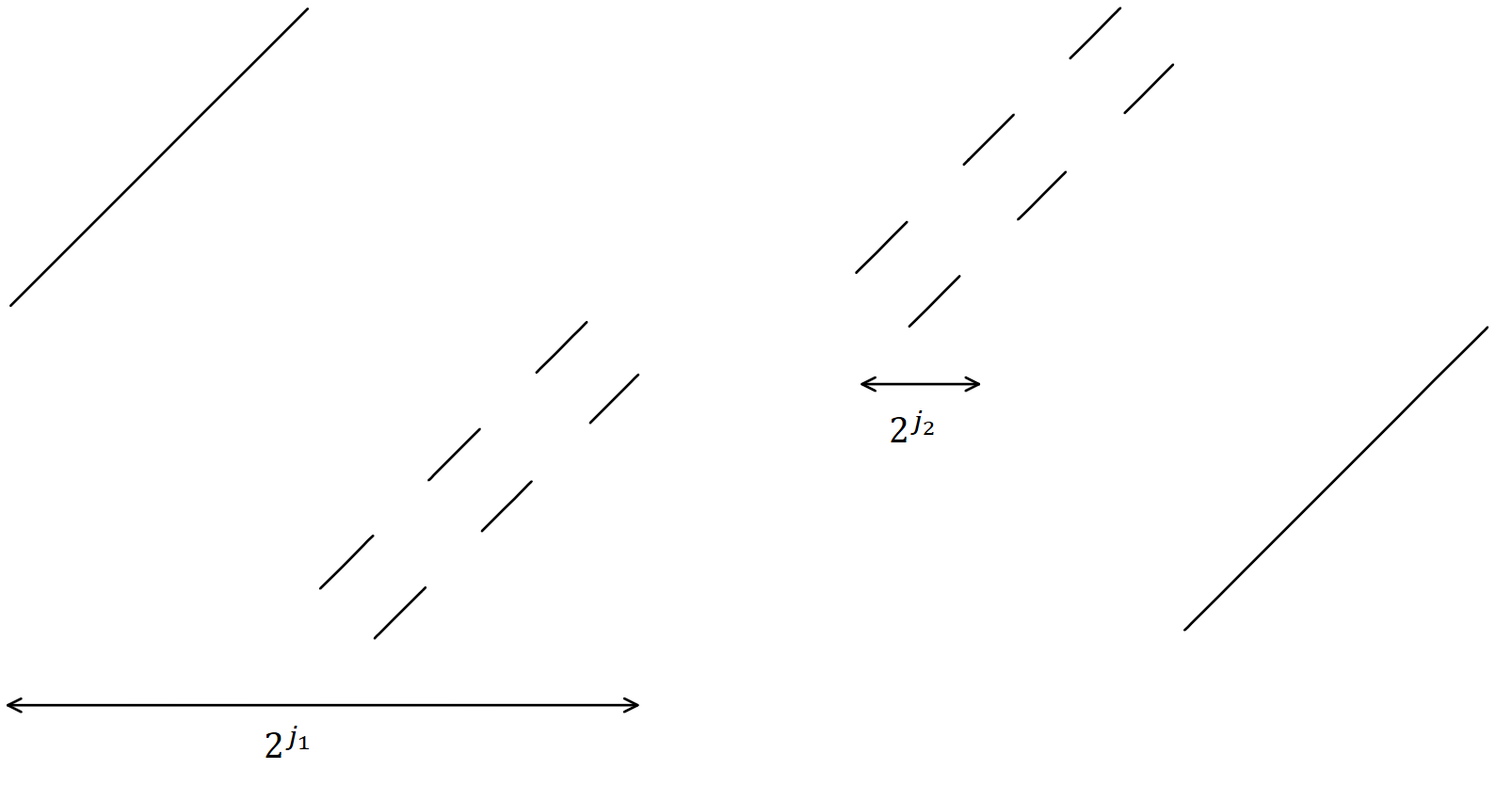}~
  \end{center}
  \caption{The relative values in $j_1$-blocks of $\mathcal{D}_1$ either look like the left or right diagrams above, with equal probability.}
  \label{fig:D1block}
\end{figure} 

We bound the distance to monotonicity of the arrays sampled from the distributions in Lemma~\ref{lem:distance-to-monotonicity-log-squared}.

\ignore{
\begin{lemma}\label{lem:distance-to-monotonicity-log-squared}
The distance to monotonicity of every array sampled from $\mathcal{D}_0$
is, with probability $1 - \frac{1}{2^{2h+3}}$, within $\frac{5}{8} \pm \frac{1}{2^{2h+2}}$. 
The distance to
monotonicity of every array sampled from $\mathcal{D}_1$ is equal to
$1/2$.
\end{lemma}
}

\begin{lemma}\label{lem:distance-to-monotonicity-log-squared}
The distance to monotonicity of every array sampled from $\mathcal{D}_0$
is, with probability $1 - \delta/2$, within $\frac{5}{8} \pm \delta$, where $\delta \le \frac{1}{2^6}$. 
The distance to
monotonicity of every array sampled from $\mathcal{D}_1$ is equal to
$1/2$.
\end{lemma}
\begin{proof}
Consider an array $A$ sampled from one of the distributions. 
Let $j_1 > j_2$ be the scales used. 
First, observe that there are no violations to monotonicity across $j_1$-blocks. 
Therefore, it is enough to focus on repairing individual $j_1$-blocks and making them monotone (without inducing new violations). 
Consider a $j_1$-block $B_\ell$ for $\ell \in [n/2^{j_1}]$.

Assume that $A$ is constructed from $\mathcal{D}_0$, and that $B_\ell$
is such that the values in the left and right halves of every
$j_2$-block in $B_\ell$ are swapped (happens with probability
$\frac{1}{2}$ for that block). 
Then we need to modify the values of at
least $3/4$ fraction of indices in $B_\ell$ to make it monotone,
since $B_\ell$ contains an exact cover by disjoint decreasing
subsequences, each of size $4$. 
Further, it is easy to see that by correcting a $3/4$ fraction of
indices we can make $B_\ell$ monotone.
In case swapping of values is done for none of the $j_2$-blocks in $B_\ell$ (happens with probability $\frac{1}{2}$), then we can repair $B_\ell$ if and only if we modify the values on half the indices in $B_\ell$. 
Therefore, the expected distance to monotonicity of each block $B_\ell$, and thereby, of $A$ is equal to $5/8$. 
Note that the specific values of the scales did not matter in the above argument.

Let $\delta = 1/2^6$. 
We now show that the distance of $A$ from monotonicity is $\frac{5}{8}
\pm \delta$, with high probability. 
For a block $B_\ell, \ell \in [n/2^{j_1}]$, let $\dist(B_\ell)$ denote the distance of the
block from monotonicity, normalized by the block length $2^{j_1}$.
We can see that the random variable $\frac{\sum_{\ell} \dist(B_\ell)}{n/2^{j_1}}$ corresponds
to the normalized Hamming distance of $A$ from monotonicity.
By Hoeffding's bound, we have:
\begin{align*}
\Pr\left[\left|\frac{\sum_{\ell} \dist(B_\ell)}{n/2^{j_1}} - \frac{5}{8}\right| > \frac{1}{2^{6}}\right] \le \frac{2}{\exp(8)} \leq \frac{1}{2^{7}} = \frac{\delta}{2}.
\end{align*}

\ignore{Indeed, let $S$ be the
number of $j_1$-blocks for which the swap of values happen inside
every $j_2$-block in $B_\ell, \ell \in [n/2^{j_1}]$ (namely the block
is as depicted in the right-hand side of Figure~\ref{fig:D0block}). Since
such a swap occurs with probability $1/2$, by a Chernoff bound, we conclude
that, with probability at least $1-2\exp(-\delta^2 \cdot \mu/3)$, the value $S 
\in [\mu(1-\delta), \mu(1+ \delta)],$ for any small enough $\delta$
and $\mu = n/2^{j_1 + 1}$. Setting $\delta = \frac{1}{20}$ and using the
bound $j_1 < \log n - 2\log\log n$ implies the result.}

\ignore{
For $\ell \in [n/2^{j_1}]$, let $X_\ell$ denote the indicator random variable for whether or not the swap of values happen inside every $j_2$-block in $B_\ell$.
We know that $\E[X_\ell] = \frac{1}{2}$.
Let $\mu$ be shorthand for $\E[\sum_{\ell \in [n/2^{j_1}]} X_\ell] = \frac{n}{2^{j_1+1}} > \frac{\log^2 n}{2}$.
By a Chernoff bound, 
\begin{align*}
\Pr[\sum_{\ell \in [n/2^{j_1}]}X_\ell \notin [(1-\delta)\cdot \mu, (1+\delta)\cdot\mu]] \le 2\exp(-\delta^2 \cdot \mu) < 2\exp(-\delta^2 \log^2 (n) /2).
\end{align*} 

Therefore, with probability at least $1 - 2\exp(-\delta^2 \log^2 (n) /2)$, the distance of $A$ from monotonicity is within $\frac{5}{8} \pm \frac{\delta}{8}$.
We obtain the desired bound by setting $\delta = \frac{1}{20}$.
}

Assume now that $A$ is constructed from $\mathcal{D}_1$.
For $\ell \in [n/2^{j_1}]$, if the swap of values happens within every $j_2$-block in the left half of $B_\ell$, then we can repair $B_\ell$ by setting every value in the left half of that block to the smallest value in the right half of the same block.
An analogous repair can be done if the swap happens in the right hand
side of the block.
In both cases, we only change values of at most half the number of indices in each block. 
The reader can easily convince themselves that at least half the values per block
need to be changed to make a $j_1$-block monotone, which concludes the proof for the given scales.
As before, the argument is independent of the choice of the scales. 
\end{proof}

We now show that every nonadaptive comparison-based deterministic algorithm distinguishing
$\mathcal{D}_0$ and $\mathcal{D}_1$ with high probability needs to
make $\Omega(\log^2 n)$ queries.
Lemma~\ref{lem:log-squared-lowerbound-indistinguishability} argues this and completes the proof of our result.
\begin{lemma}\label{lem:log-squared-lowerbound-indistinguishability}
  Every comparison-based nonadaptive deterministic algorithm that, with probability at least $2/3$, distinguishes the distributions $\mathcal{D}_0$ and $\mathcal{D}_1$ has to make $\Omega(\log^2 n)$ queries.
\end{lemma}
\begin{proof}
  Consider an arbitrary deterministic comparison-based nonadaptive algorithm $T$ that makes $o(\log^2 n)$ queries and aims to distinguish $\mathcal{D}_0$ and $\mathcal{D}_1$.
  Let $Q \subseteq [n]$ denote the set of queries that $T$ makes.
  
  Consider an array $A$ sampled according to one of the distributions.
  Recall that $\mathcal{T}$ denotes an ordered binary tree whose leaves are the indices of $A$. 
  Let $j_1,j_2 \in [\log n]$ such that $j_2 < j_1$ denote the scales used while constructing $A$.  
  Let $E$ denote the (bad) event that $Q$ contains four indices
        $w < x < y <z \in [n]$ such that for some $\ell \in
        [n/2^{j_1}]$, (1) Each of $\{w,x,y,z\}$ belongs to the same $\ell$-th $j_1$-block, (2) $\mathsf{LCA}(w,x)$ and $\mathsf{LCA}(y,z)$ are both nodes in $\mathcal{T}$ at level $j_2$, and (3) $\mathsf{LCA}(x,y)$ is at level $j_1$.
  By an argument of Ben-Eliezer, Canonne, Letzter, and Waingarten~\cite{BCLW19}, the probability, over the choice of the scales $j_1, j_2$, of the event $E$ is at most $1/3$. 
  In the rest of the proof, we fix the scales $(j_1, j_2)$ for which $E$ does not happen.
  
  Let $x,y \in Q$ be such that $\mathsf{LCA}(x,y)$ is at level $j_2$ in $\mathcal{T}$.
  In the rest of the proof, for simplicity, we refer to such
        queries as being \emph{$j_2$-cousins}.  Let $B$ be a
        $j_1$-block, and let $Q_L(B) $ be the queries in $Q$ that are
        in the left half of $B$, and $Q_R(B)$ the queries in $Q$ that
        are in the right half of $B$.
  By our conditioning, for each $j_1$-block $B$, either all
        $j_2$-cousins in $B$ belong to $Q_L(B)$ or to $Q_R(B)$  but not both.
          Consider the half of $B$ that does not contain any $j_2$-cousins.
  In the algorithm's view, the array values in that half are increasing, whether $A$ is sampled from $\mathcal{D}_0$ or $\mathcal{D}_1$.

  Assume that $A$ is sampled from $\mathcal{D}_0$.
  We show that there is a way to sample an array $A'$ from $\mathcal{D}_1$ such that the algorithm's view on $A$ and $A'$ are identical.
  The scales of $A'$ have to be identical to those of $A$. 
  We only need to specify how swapping of values is done inside the $j_1$-blocks, as part of constructing $A'$.
  
  Note that we only need to consider the $j_1$-blocks in which at least two queries fall.
  Consider such a block $B$, and assume that $Q_R(B)$ contains no
        $j_2$-cousins. 
  Consider the case that swapping of values was done within every $j_2$-block inside the block $B$ while constructing $A$.
  Then, in $A'$, we swap the values only within the $j_2$-blocks inside the left half of $B$.
  In the other case that no swapping of values (within $j_2$-blocks) was done while constructing $A$, we swap the values only within the $j_2$-blocks inside the right half of $B$. 
  It is easy to see that the relative values within block $B$ in the array $A'$ are consistent with that of an array sampled from $\mathcal{D}_1$.     
  One can make similar arguments about coupling the arrays $A$ and $A'$ on blocks $B$ such that the only occurrences of indices $x,y \in Q \cap B$ that are $j_2$-cousins are in the right half of $B$. 
  
  We conclude that for any scales $j_1, j_2$ for which $E$
        does not happen, the view of the algorithm  making queries $Q$ is
        identical on $A'$ and $A$. Hence the algorithm  cannot distinguish
        between the case that the array is sampled from $\mathcal{D}_0$ or from
        $\mathcal{D}_1$ for such scales. As this is true for any scales for
        which $E$ does not happen, this concludes the proof.
        
  Observe that the only place in the analysis where we made use of the bound on the number of queries is in arguing that the event $E$ happens with low probability. 
\end{proof}

\subsection{A $\log^{\omega(1)} n$ Lower Bound}\label{sec:super-polylog-lb}
In what follows we strengthen the $\Omega(\log^2 n)$ lower bound and prove Theorem~\ref{thm:adaptive-tolerant-lb}. 
We make use of the
idea of Ben-Eliezer, Canonne, Letzter, and Waingarten~\cite{BCLW19} for lower bounding query complexity of nonadaptive
detection of larger forbidden order patterns. The idea is to use more
than just two scales. This, in turn, makes the difference between the
distances of arrays sampled from the two distributions smaller,
which is why we only get a $\log^{\omega(1)} n$ lower bound.

Throughout this section, we assume that $n$ is a power of $2$ and that all logarithms are base $2$.
Let $h \ge 2$ be an integer parameter. 
We first describe the two distributions $\mathcal{D}_0^{(h)}$ and
$\mathcal{D}_1^{(h)}$ on real-valued arrays, such that
no comparison-based deterministic nonadaptive algorithm that makes $o(\log^h n)$
queries can distinguish between the 
distributions with high probability. Further, the distance to
monotonicity of each array sampled from $\mathcal{D}_0^{(h)}$ will be
significantly different than that of arrays sampled from $\mathcal{D}_1^{(h)}$.

The distributions are defined recursively.
For the base case $h = 2$, let $\mathcal{D}_0^{(2)}$ and $\mathcal{D}_1^{(2)}$ be equal to 
$\mathcal{D}_0$ and $\mathcal{D}_1$ defined with scales $j_1, j_2$ as in Section~\ref{sec:log-squared-lower-bound}, respectively.

For the general case, we begin with an array 
$A$ of length $n$ such that $A[u] = u$ for all $u \in [n]$. 
Then we sample scales $j_1, j_2, \dots, j_h \in [\log (n) - \log (h+2) - (4h+4)]$ such that $j_1 > j_2 > \dots > j_h$ and that for every $i \in [h-1]$, we have $j_{i+1} \le j_i - \log (h+2) - (4h+4)$.
For each $\ell \in [n/2^{j_1}]$, we swap the values in the left and right halves of the $j_1$-block $B_\ell$. 
\begin{itemize}
  \item {\bf Distribution} $\mathcal{D}_0^{(h)}$: 
  Independently for each $\ell \in [n/2^{j_1}]$:
  \begin{enumerate}
  \item with probability $1/2$, rearrange the values within the left half of $B_\ell$ as if we are sampling an array of length $2^{j_1-1}$ from $\mathcal{D}_0^{(h-1)}$, and rearrange the values 
  within the right half of $B_\ell$ as if we are sampling an
        array of length $2^{j_1-1}$ from $\mathcal{D}_0^{(h-1)}$, where
        the scales used in both cases are $j_2, j_3, \dots, j_h$. We
        refer to such a $j_1$-block as a $00$-block.
  \item with the remaining probability $1/2$, rearrange the values within the left half of $B_\ell$ as if we are sampling an array of length $2^{j_1-1}$ from $\mathcal{D}_1^{(h-1)}$, and rearrange the values 
  within the right half of $B_\ell$ as if we are sampling an
        array of length $2^{j_1-1}$ from $\mathcal{D}_1^{(h-1)}$, where
        the scales used are $j_2, j_3, \dots, j_h$.  We
        refer to such a $j_1$-block as a $11$-block. 
  \end{enumerate}
  
  \item {\bf Distribution} $\mathcal{D}_1^{(h)}$: Independently for each $\ell \in [n/2^{j_1}]$: 
  \begin{enumerate}
  \item with probability $1/2$, rearrange the values within the left half of $B_\ell$ as if we are sampling an array of length $2^{j_1-1}$ from $\mathcal{D}_0^{(h-1)}$, and rearrange the values 
  within the right half of $B_\ell$ as if we are sampling an
        array of length $2^{j_1-1}$ from $\mathcal{D}_1^{(h-1)}$, where
        the scales used in both cases are $j_2, j_3, \dots, j_h$.  We
        refer to such a $j_1$-block as a $01$-block.
  \item with the remaining probability $1/2$, rearrange the values within the left half of $B_\ell$ as if we are sampling an array of length $2^{j_1-1}$ from $\mathcal{D}_1^{(h-1)}$, and rearrange the values 
  within the right half of $B_\ell$ as if we are sampling an
        array of length $2^{j_1-1}$ from $\mathcal{D}_0^{(h-1)}$, where
        the scales used in both cases are $j_2, j_3, \dots, j_h$.  We
        refer to such a $j_1$-block as a $10$-block.
  \end{enumerate}
\end{itemize}
Thus, conditioned on the scales being $j_1, \ldots , j_h$,  the following is the
difference between $\mathcal{D}_0^{(h)}$ and $\mathcal{D}_1^{(h)}$. 
In
every $j_1$-block $B$, the distribution $\mathcal{D}_0^{(h)}$ induces identical
distributions on the first and second halves of $B$, where these are
either both $\mathcal{D}_0^{(h-1)}$ or both $\mathcal{D}_1^{(h-1)}$ (each with
probability $1/2$; i.e., 00 or 11 blocks). The distributions that $\mathcal{D}_1^{(h)}$ induces
 on the first and second halves of $B$ 
are either $(\mathcal{D}_0^{(h-1)}, \mathcal{D}_1^{(h-1)})$
respectively or $(\mathcal{D}_1^{(h-1)}, \mathcal{D}_0^{(h-1)})$ (each
with probability $1/2$; i.e., 01 or 10 blocks).


      \medskip 
 \begin{lemma}\label{lem:dist-lowerbound-general}
With probability at least $1 - \frac{1}{2^{2h+3}}$, the distance to monotonicity of an array sampled
from $\mathcal{D}_1^{(h)}$ is $d^{(h)}_1  = 1 - \frac{1}{2^{h-1}} \pm \frac{3}{2^{2h + 2}}$, and with
high probability the distance to monotonicity of an array sampled
from $\mathcal{D}_0^{(h)}$ is  $d^{(h)}_0 = 1 - \frac{1}{2^{h-1}} + \frac{1}{2^{2h-1}} \pm \frac{3}{2^{2h + 2}}$. 
      \end{lemma}
\begin{proof}
Let $j_1, j_2, \dots, j_h$ denote the scales used to sample from the distributions. Note that these are random variables.

\begin{claim}\label{clm:lower-bound-induction}
Let $s \in [h]$. Let $d^{(s)}_b$ denote the (normalized Hamming) distance to monotonicity of an array sampled from $\mathcal{D}_b^{(s)}$ for $b \in \{0,1\}$.
Let $\alpha_s = \frac{1}{2^{2h - s + 2}}$, $\gamma_s = \frac{1}{2^s}, \delta_s = \frac{1}{2^{2h-s+3}},$ and $p_s = \frac{1}{2^{2h+4}}$.
Then we have:
\begin{align}
\E[d_1^{(s)}] &= 1 - \frac{1}{2^{s-1}} \pm \frac{\alpha_s}{2^{s-1}} \label{eqn:a}\\
\E[d_0^{(s)}] &= 1 - \frac{1}{2^{s-1}} + \frac{\gamma_s}{2^{s-1}} \pm \frac{\alpha_s}{2^{s-1}}. \label{eqn:b}
\end{align}
Moreover,
\begin{align}
\Pr\left[d_1^{(s)} \notin \E[d_1^{(s)}] \pm \frac{\delta_{s}}{2^{s-1}}\right] &\leq p_s \label{eqn:c}\\
\Pr\left[d_0^{(s)} \notin \E[d_0^{(s)}] \pm \frac{\delta_{s}}{2^{s-1}}\right] &\leq p_s. \label{eqn:d}
\end{align}
\end{claim}
\begin{proof}
We the prove the claim by induction on $s \in [h]$.

\noindent{\bf Base Case} $s = 2$:
The proof is identical to the proof of Lemma~\ref{lem:distance-to-monotonicity-log-squared}.
\smallskip

\noindent Assume that the claim holds for $s \leq t - 1 < h$.
\smallskip

\noindent{\bf Inductive Step} $s = t$:
Let $A_b^{(t)}$ denote an array of length $n$ sampled from the distribution $\mathcal{D}_b^{(t)}$ for $b \in \{0,1\}$.
For a contiguous subarray $A'$ of $A_b^{(t)}$, let $\dist(A')$ denote the normalized Hamming distance of $A'$ to monotonicity, where the normalization factor is $|A'|$.
Let $B_\ell$ denote the $\ell$-th $j_1$-block of an $n$-length array, where $\ell \in [n/2^{j_1}]$.
Let $L_\ell$ and $R_\ell$ denote the left and right halves of $B_\ell$, respectively.
It can be observed that 
\begin{align}
\dist(B_\ell) = \frac{1}{2} + \frac{1}{2} \cdot \min \{\dist(L_\ell),\dist(R_\ell)\}~\mbox{for all}~\ell. \label{eqn:e}
\end{align}

\noindent {\it Proving Equations~\ref{eqn:a} and~\ref{eqn:c} for $s = t$}:
Consider an array $A_1^{(t)}$.
We can see that, $d_1^{(t)} = \frac{\sum_{\ell} \dist(B_\ell)}{n/2^{j_1}}$.
In order to bound $\E[d_1^{(t)}]$, it is enough to bound $\E[\dist(B_\ell)]$ for an arbitrary $\ell$ and then use the linearity of expectation. We now bound $\E[\dist(B_\ell)]$ from both above and below. 
Note that with probability $1/2$, we have $L_\ell \sim \D{0}{t-1}, R_\ell \sim \D{1}{t-1}$ and with probability $1/2$, we have $L_\ell \sim \D{1}{t-1}, R_\ell \sim \D{0}{t-1}$.
Since $\dist(B_\ell)$ is identically distributed for all $\ell$, we have 
\begin{align}
\E[\dist(B_\ell)] = \E[\di{1}{t}].\label{eqn:f}
\end{align}

By the inductive hypothesis and a union bound, we know that, with probability at least $1 - 2p_{t-1}$:
\begin{align*}
\dist(L_\ell) \in \E\left[\dist(L_\ell)\right] \pm \frac{\delta_{t-1}}{2^{t-2}},~\mbox{and}~\dist(R_\ell) \in \E\left[\dist(R_\ell)\right] \pm \frac{\delta_{t-1}}{2^{t-2}}.
\end{align*}
That is, with probability at least $1 - 2p_{t-1}$:
\begin{align}
\min\{\dist(L_\ell),\dist(R_\ell)\} \in \E[\di{1}{t-1}] \pm \frac{\delta_{t-1}}{2^{t-2}}. \label{eqn:g}
\end{align}
The above holds because, by our assumption, $\E[\di{1}{t-1}] + \frac{\delta_{t-1}}{2^{t-2}} < \E[\di{0}{t-1}] - \frac{\delta_{t-1}}{2^{t-2}}$.

Let $I$ be a shorthand for $\E[\min\{\dist(L_\ell),\dist(R_\ell)\}]$.
From all of the above, we have:
\begin{align*}
(1 - 2p_{t-1})\cdot (\E[\di{1}{t-1}] - \frac{\delta_{t-1}}{2^{t-2}}) \leq I \leq \E[\di{1}{t-1}] + \frac{\delta_{t-1}}{2^{t-2}} + 2p_{t-1}.
\end{align*} 
Using the above inequalities, and Equation~\ref{eqn:e}, we get upper and lower bounds on $\E[\di{1}{t}]$ in terms of $\E[\di{1}{t-1}]$.
\begin{align*}
\E[\di{1}{t}] &\leq \frac{1}{2} + \frac{1}{2} \cdot \left(\E[\di{1}{t-1}] + \frac{\delta_{t-1}}{2^{t-2}} + 2p_{t-1}\right)\\
\E[\di{1}{t}] &\geq \frac{1}{2} + \frac{1}{2} \cdot \left((1 - 2p_{t-1})\cdot (\E[\di{1}{t-1}] - \frac{\delta_{t-1}}{2^{t-2}})\right)
\end{align*}

By further using the conditions on $\E[\di{1}{t-1}]$ by the induction hypothesis, we get the following requirements, which are all satisfied by our setting of parameters:
\begin{align*}
\delta_{t-1} + \alpha_{t-1} + 2^{t}\cdot p_{t-1} &\leq \alpha_t\\
\alpha_{t-1} + \delta_{t-1} + 2^{t-1}p_{t-1} &\leq \alpha_t + 2p_{t-1} \cdot (1 + \alpha_{t-1} + \delta_{t-1}).
\end{align*}
The above finishes the proof of Equation~\ref{eqn:a} in the inductive step.
Note that, the second inequality above is redundant, as it is already implied by the first one.

We now prove Equation~\ref{eqn:c} for $s = t$. 
By Hoeffding's inequality, in association with our assumption that $j_1 < \log (n) - \log (h+2) - (4h+4)$, 
\begin{align*}
\Pr\left[\left|\frac{\sum_{\ell} \dist(B_\ell)}{n/2^{j_1}} - \E[\di{1}{t}]\right| > \frac{\delta_t}{2^{t-1}}\right] \leq 2\exp(-\frac{2n}{2^{j_1}}\cdot \frac{\delta_t^2}{2^{2t-2}}) = \frac{2}{\exp(2h+4)} \leq p_t.
\end{align*}

\noindent\textit{Proving Equations~\ref{eqn:b} and~\ref{eqn:d} for $s=t$}:
Consider an array $A_0^{(t)}$. 
For $\ell \in [n/2^{j_1}]$, we have that, with probability $1/2$, the subarrays $L_\ell, R_\ell \sim \D{0}{t-1}$, and with probability $1/2$, the subarrays $L_\ell, R_\ell \sim \D{1}{t-1}$. 
As before, 
$$\E[\di{0}{t}] = \E[\dist(B_\ell)],$$
for arbitrary $\ell \in [n/2^{j_1}]$.
Using earlier notation, $I = \E[\min\{\dist(L_\ell), \dist(R_\ell)\}]$.

We have
\begin{align*}
I &= \frac{1}{2} \cdot \E[\min\{\dist(L_\ell), \dist(R_\ell)\}|L_\ell, R_\ell \sim \D{0}{t-1}] + \frac{1}{2} \cdot \E[\min\{\dist(L_\ell), \dist(R_\ell)\}|L_\ell, R_\ell \sim \D{1}{t-1}]\\
& = \frac{1}{2} \cdot \E[\di{0}{t-1}] + \frac{1}{2} \cdot \E[\di{1}{t-1}]
\end{align*} 

Thus, using the bounds on $\E[\di{0}{t-1}]$ and $\E[\di{1}{t-1}]$ obtained via the induction hypothesis, we can bound $\E[\di{0}{t}] = \E[\dist(B_\ell)] = \frac{1}{2} + \frac{I}{2}$ as:
\begin{align*}
\E[\di{0}{t}] &\leq \frac{1}{2} + \frac{1}{4} [1 - \frac{1}{2^{t-2}} + \frac{\gamma_{t-1}}{2^{t-2}} + \frac{\alpha_{t-1}}{2^{t-2}} ] + \frac{1}{4} [ 1 - \frac{1}{2^{t-2}} + \frac{\alpha_{t-1}}{2^{t-2}}]\\
\E[\di{0}{t}] &\geq \frac{1}{2} + \frac{1}{4} [1 - \frac{1}{2^{t-2}} + \frac{\gamma_{t-1}}{2^{t-2}} - \frac{\alpha_{t-1}}{2^{t-2}} ] + \frac{1}{4} [ 1 - \frac{1}{2^{t-2}} - \frac{\alpha_{t-1}}{2^{t-2}}]
\end{align*}
Imposing the conditions necessary to prove the induction statement, we get the following inequalities:
\begin{align*}
\gamma_{t-1} + 2\alpha_{t-1} &\leq 2\alpha_{t} + 2\gamma_t\\
2\gamma_t + 2\alpha_{t-1} &\leq \gamma_{t-1} + 2\alpha_t 
\end{align*}

The above are satisfied by our setting of parameters since we have that $\alpha_{t-1} \le \alpha_t$ and $2\gamma_t = \gamma_{t-1}$.

The proof of Equation~\ref{eqn:d} for $s = t$ is identical to that of the proof of Equation~\ref{eqn:c} and is a simple application of Hoeffding's inequality. 
\end{proof}
The lemma follows by applying the Claim~\ref{clm:lower-bound-induction} to the case that $s = h$.
\end{proof}
\ignore{
\noindent
{\em Proof idea.}  The proof is by induction and is not presented
here. The top level idea is as follows.
We
prove the lemma for each tuple of scales $(j_1, j_2, \dots, j_t),$ used in sampling from the
distributions $\mathcal{D}_0^{(t)}$ and $\mathcal{D}_1^{(t)}$. In each
distribution, the distance is the average distance on the $n/2^{j_1 +
  1}$  subarrays of size $2^{j_1+1}$.  Let $B$ be such a
subarray. Then either $B$ is a $(00)$ or $(11)$ block if sampled from
$\mathcal{D}_0^h$, each with probability $1/2$, or if sampled from
$\mathcal{D}_1^h$ it is a $(01)$ or a $(10)$ block, each w.p $1/2$. In
all cases, all values in the first half of $B$ are larger than all values
in the second half of $B$. Thus it is easy to see that the distance in
$B$ is $1/2 + d'/2$ where $d'$ is the distance of either the left or
right half of $B$, whichever is smaller. But these
corresponding parts of $B$  are distributed according to $\mathcal{D}_1^{h-1}$
or $\mathcal{D}_0^{h-1}$. We then bound the expected distance. We
avoid further details. \qed }

\begin{lemma}\label{lem:correctness-lowerbound-general}
Every deterministic comparison-based algorithm  needs $\Omega(\log^h n)$ queries to distinguish $\mathcal{D}_0^{(h)}$ and $\mathcal{D}_1^{(h)}$, for all positive integer $h \ge 2$.
\end{lemma}
\begin{proof}
We prove the statement by induction. The idea, as in the case of
$h=2$, is to condition on some high probability event that depends only on the scales,
which are chosen identically for $\mathcal{D}_0^{(h)}$ and
$\mathcal{D}_1^{(h)}$. 
Under such a conditioning, the
distributions induce identical distributions on any fixed set of
queries. This in turn will be shown by a coupling argument as done for
$h=2$. The details follow.

For the base case $h = 2, $ the statement is proved in Section~\ref{sec:log-squared-lower-bound}.
Assume that the statement is true for all $h \le t - 1$ and prove the statement 
for  $h = t$. 
Consider a deterministic nonadaptive algorithm  $T$ whose set of queries
is $Q$, with $|Q| = o(\log^t n)$. 

We first define inductively, what we mean by a set of queries to be $h$-\emph{bad}.
For $h = 2$, let $j_1, j_2$ be the scales used in sampling.
A set of $4$ queries $\{w, x, y, z\} \subseteq Q$ is $2$-bad 
 if they satisfy the following. 
There exists a $j_1$-block $B$ such that (1) all the four indices are
in $B$, (2) $w,x$ belong to the same $j_2$-block in the left half of
$B$, and (3) $y,z$ belong to the same $j_2$-block in the right half of
$B$ (this is equivalent to being $j_2$-cousins from
Section~\ref{sec:log-squared-lower-bound}). 

Assume that we have defined a set of $2^{t-1}$ queries to be
$(t-1)$-bad. 
We now perform the inductive step for $h = t$.
Recall, at first, that the left and right halves of each $j_1$-block in an array sampled from either $\mathcal{D}_0^{(t)}$ or $\mathcal{D}_1^{(t)}$ look like arrays sampled from $\mathcal{D}_0^{(t-1)}$ or $\mathcal{D}_1^{(t-1)}$.
Let $j_1, j_2, \dots, j_t$ be the scales used in sampling from $\mathcal{D}_0^{(t)}$ and $\mathcal{D}_1^{(t)}$. 
Consider a set of $k = 2^t$ queries $\{x_1, x_2, \dots x_k\} \subseteq Q$. 
This set is $t$-bad 
 if it satisfies the following.
There exists a $j_1$-block $B$ such that (1) the indices $\{x_1, x_2, \dots x_{k/2}\}$ belong to the left half of $B$, (2) the indices $\{x_{k/2 + 1}, x_{k/2 + 2}, \dots x_{k}\}$ belong to the right half of $B$, and (3) the sets $\{x_1, x_2, \dots x_{k/2}\}$ and $\{x_{k/2 + 1}, x_{k/2 + 2}, \dots x_{k}\}$ are $(t-1)$-bad with respect to distributions $\mathcal{D}_0^{(t-1)}$ and $\mathcal{D}_1^{(t-1)}$ (sampled with scales $j_2, j_3, \dots, j_t$).

Let $E_t$ be the bad event that there exists a $t$-bad set of queries with respect to $\mathcal{D}_0^{(t)}$ or $\mathcal{D}_1^{(t)}$ in $Q$.  
Ben-Eliezer, Canonne, Letzter, and Waingarten~\cite{BCLW19} show that for $|Q| = o(\log^t n)$, the probability, over the choice of the scales $j_1, j_2, \dots, j_t$, of the event $E_t$ is at most $\frac{1}{3}$. 

\begin{claim}
Conditioned on $E_h$ not happening, the views of every deterministic comparison-based algorithm is identical on arrays sampled from $\mathcal{D}_0^{(h)}$ and $\mathcal{D}_1^{(h)}$.
\end{claim}
\begin{proof}
We prove this by induction.
The base case $h = 2$ comes from the proof of Lemma~\ref{lem:log-squared-lowerbound-indistinguishability}.
Assume that the claim holds true for $h \le t-1$.
We now prove the claim for the case that $h = t$.

Consider an array $A$ sampled from $\mathcal{D}_0^{(t)}$ with scales $j_1, j_2, \dots, j_t$.
Conditioned on $E_t$ not happening, for every $j_1$-block $B$, either the
left or right halves of $B$ does not contain any $(t-1)$-bad set of queries.
In other words, the event $E_{t-1}$ does not happen in either the left half or the right half of $B$. 
Recall that the left and right halves of $B$ are themselves arrays sampled from $\mathcal{D}_0^{(t-1)}$ or $\mathcal{D}_1^{(t-1)}$ with scales $j_2, \dots, j_t$. 

We now show a way to sample an array $A'$ from $\mathcal{D}_1^{(t)}$ such that the view of the algorithm  on both $A$ and $A'$ are identical.
That is, we specify whether each $j_1$-block of $A'$ is either a $01$-block or a $10$-block.

Consider a $j_1$-block $B$. 
If $E_{t-1}$ happens neither in the left nor in the right half of $B$, then we can associate $B$ in $A'$ with a $01$-block or a $10$-block uniformly at random. 
By the inductive hypothesis, no algorithm  can determine which distribution was the halves are sampled from.
Hence, the view of the algorithm  on $B$ is identical whether the array is $A$ or $A'$.

Assume now that $E_{t-1}$ happens only in the left half of $B$ and that $B$ is a $00$-block in $A$. 
Then we associate $B$ in $A'$ with a $01$-block. 
In case $B$ is a $11$-block in $A$, we associate $B$ in $A'$ with a $10$-block. 
Note that our association is consistent with how a block is sampled when an array is sampled from $\mathcal{D}^{(t)}_1$, since $B$ being a $01$-block or a $10$-block in $A'$ is equally likely and is independent of other blocks.
Since the algorithm  cannot determine whether the right half of $B$ is sampled from $\mathcal{D}_0^{(t-1)}$ or $\mathcal{D}_1^{(t-1)}$, the view of the algorithm  on $B$ is identical whether the array is $A$ or $A'$.
Similarly, we can also show an association for the case that $E_{t-1}$ happens only in the right block of $B$. 
This completes the proof of the claim.
\end{proof}
Note that the only place in the proof where we used the bound on the number of queries was in bounding the probability of occurrence of a bad set of indices.
\end{proof} 

      We note that, since for constant $h$, the distance from monotonicity of arrays sampled from both distributions
      are constants, and that these distances (and therefore, LIS lengths) differ by $2^{-\Theta(h)}$. Lemmas~\ref{lem:dist-lowerbound-general} and~\ref{lem:correctness-lowerbound-general} together directly imply Theorem~\ref{thm:adaptive-tolerant-lb}.

\section{Nonadaptive Erasure-Resilient monotonicity Tester}\label{sec:na-er-monotonicity-tester}

In this section, we prove
Theorem~\ref{thm:nonadaptive-erasure-resilient-monotonicity-test}. That
is, we show that there exists a nonadaptive erasure-resilient tester for monotonicity of real-valued arrays.
We present a basic tester in
Algorithm~\ref{alg:nonadaptive-monotonicity-tester}; that is, it is
a nonadaptive $1$-sided error tester for monotonicity with a relatively
small probability of success. Then the probability of success can be
amplified to the standard value of $2/3$, by repeating
Algorithm~\ref{alg:nonadaptive-monotonicity-tester} for $\left\lceil\frac{120}{\epsilon^2}\right\rceil$
times.
Algorithm~\ref{alg:nonadaptive-monotonicity-tester} gets as input, a parameter $\epsilon \in (0,1)$, and has oracle access to a real-valued array $A$ that has at most an $\alpha$ fraction of its values erased, where $\alpha < 1 - \epsilon$.
We point out that our tester does not require knowledge of $\alpha$. 

The intuition behind the algorithm is quite simple. As a first strategy, one can sample a
random index $s \in [n]$, query it, and run a randomized
binary search for the value $f(s)$, provided that $f(s)$ is not
erased. One rejects, if a violation to monotonicity is found among the queried nonerased values. Such an
algorithm is obviously a $1$-sided error algorithm. Let $A$ be $\epsilon$-far from
monotone. If $A$ contained no
erasures, this algorithm is known
to detect a violation to monotonicity with high probability. Further, the expected
depth of a randomized binary search tree, which is also the expected query complexity of this algorithm, is $O(\log n)$.
The problem with having 
erasures is that once a pivot in the search is erased, the search does
not have any indication as to what direction it should go, and such an
event will occur with high probability when there are a constant fraction of
erasures. Instead, we search for $s$ rather than for $f(s)$ (which
obviously will end in finding $s$ at the end), and hope to find
a violation to monotonicity among the values on the set of points that are sampled. We then
show that this tester has a good success probability using a lemma from
\cite{NewmanRRS19}.
The details follow.


\begin{algorithm}
  \caption{Nonadaptive erasure-resilient monotonicity tester}
  \begin{algorithmic}[1]
    \Require parameter $\epsilon \in (0,1)$, oracle access to an $\alpha$-erased array $A$ of length $n$, where $\alpha \le 1 - \epsilon$ 
    \State Sample $s \in [n]$ uniformly and independently at random and query $A[s]$.\label{step:search-point}
    \State $\ell \gets 1, r \gets n$.
    \Repeat 
    \State Sample a pivot $p \in [\ell,r]$ uniformly and independently at random and query $A[p]$. 
    \If {$s \in [\ell, p]$}
    \State $r \gets p$
    \Else 
    \State $\ell \gets p$
    \EndIf
    \Until {$p = s$}
    \State Query the values of $A$ at $20/\epsilon$ indices to the left and to the right of $s$. \label{stp:sequence-queries}
    \State \textbf{Reject} if any pair of nonerased indices
                queried in this iteration violate monotonicity, and
                otherwise \textbf{Accept}.
  \end{algorithmic}
\label{alg:nonadaptive-monotonicity-tester}
\end{algorithm}

 To analyze the success we
need the following definitions, and a lemma from~\cite{NewmanRRS19}.
\begin{definition}\label{def:deserted-element}
  Given a set $S \subseteq [n]$ and a parameter $\gamma \in [0,1]$, an element $i \in S$ is called \emph{$\gamma$-deserted}, if there exists an interval $I \subseteq [n],$ containing $i,$ such that $|S \cap I| < \gamma |I|$. 
\end{definition}

\begin{lemma}[\cite{NewmanRRS19}]\label{lem:bound-on-deserted-elements}
  Let $S \subseteq [n]$ with $|S| \geq \delta n$. 
  For every $\gamma < 1$, at most 
  $3 \gamma (1 - \delta)  n/(1 - \gamma)$ 
  elements of $S$  are $\gamma$-deserted.
\end{lemma}
\begin{lemma}\label{lem:nonadaptive-monotonicity-rejection}
Algorithm~\ref{alg:nonadaptive-monotonicity-tester} rejects, with probability at least $\epsilon^2/60$, every $\alpha$-erased real-valued array $A$ of length $n$ that is $\epsilon$-far from monotone.
\end{lemma}
\noindent
{\bf Proof.}
Let $E \subseteq [n]$ denote the set of erased indices in $A$.
The fact that $A$ is $\epsilon$-far from monotone implies that
there is a matching of disjoint violating pairs of size at least
$|M|= \ell \geq \epsilon n/2$, where $M =\{(x_i,y_i) |~ i=1, \ldots,
\ell \}$. Let $S= \bigcup_{i \in [\ell]}
\{x_i,y_i\} \subseteq [n] \setminus E$.
By applying Lemma~\ref{lem:bound-on-deserted-elements} with
respect to $S$, $\delta = \epsilon$ and $\gamma  = \delta/10$, we get
that at most $3 \gamma (1 - \delta)  n/(1 - \gamma) \leq \frac{3
  \delta n}{10(1-\delta/10)} \leq \epsilon n/9$ indices from $S$ are
$\gamma$-deserted.  Let $S_\gamma$ be the indices in $S$ that are 
$\gamma$-deserted (with respect to $S$).

Let $M^*$ be the matching obtained by restricting $M$ onto the
indices that are not $\gamma$-deserted -- that is, $M^* =
M|_{S \setminus S_\gamma}$. Since for each deserted index in $S_\gamma$, its removal may result
in the deletion of one edge, \[|M^*| \geq \epsilon n/2 - \epsilon n /9  >
\epsilon n/3.\]

Consider an edge $(x,y) \in M^*$ such that $x < y$.
Every index $z \in S \cap [x,y]$ is such that either $(x,z)$ or $(z,y)$ violates monotonicity. 
Let $\mathsf{maj}(x,y)$ denote the element $a \in \{x,y\}$  that violates monotonicity with at least half of the indices in $S \cap [x,y]$, where we break ties arbitrarily.

Consider an execution of Algorithm~\ref{alg:nonadaptive-monotonicity-tester}.
Let $F_{[x,y]}$ denote the event that the index $\mathsf{maj}(x,y)$ is chosen as the search index $s$ in Step~\ref{step:search-point}. 

First, consider the case that $y-x \leq 1+ 20/\epsilon$. Conditioned on $F_{[x,y]}$, both $x$ and $y$ are queried in Step~\ref{stp:sequence-queries}, and hence, Algorithm~\ref{alg:nonadaptive-monotonicity-tester} rejects with probability $1$. 

Next, consider the case that $y - x > 1+ 20/\epsilon  = 1+ 2/ \gamma$. 
Conditioned on $F_{[x,y]}$, the algorithm ends in the index $\mathsf{maj}(x,y)$, and hence, the indices $x,y$ are \emph{separated by a pivot}. 
Consider the last interval $I'$ in which $x$ and $y$ were together. 
Clearly, $[x,y] \subseteq I'$.
As $x,y$ are separated, the pivot in $I'$ must have been chosen from the interval $[x,y]$.
Note that since neither $x$ nor $y$ is
$\gamma$-deserted, there are at least $\gamma (y-x) > 2$ indices in $S \cap [x,y]$.
By definition, every index in $S \cap [x,y]$ is nonerased.
Since each index in $[x,y]$ is equally likely to be chosen as a pivot, an element in $S \cap [x,y]$ that forms a violation with $\mathsf{maj}(x,y)$ is chosen as the pivot for $I'$ with probability at least $\gamma/2$. 

We conclude that the probability of Algorithm~\ref{alg:nonadaptive-monotonicity-tester} finding a violation to monotonicity is at least
\[\sum_{(x,y) \in M^*} \Pr[F_{[x,y]}] \cdot \Pr[\text{tester finds a violation}|F_{[x,y]}] \geq \frac{\epsilon}{3} \cdot
\frac{\epsilon}{20}.               ~ ~ \qed        \]  

\begin{proof}[Proof of Theorem~\ref{thm:nonadaptive-erasure-resilient-monotonicity-test}]
Algorithm~\ref{alg:nonadaptive-monotonicity-tester} clearly accepts
every monotone array with probability $1$. Further, its expected query
complexity is at most $b\log n + \frac{40}{\epsilon}$, where $b\log n$, for an absolute constant $b$,
is the expected depth a random binary tree. We can terminate and accept if
the algorithm has not detected a violation to monotonicity after, say $10(b\log n + \frac{40}{\epsilon})$ queries. 
This is done to ensure that the
worst-case query complexity is $O(\log n + \frac{1}{\epsilon})$ while maintaining the $1$-sidedness of error and the success probability as
  $\Omega(\epsilon^2)$. Repeating
  Algorithm~\ref{alg:nonadaptive-monotonicity-tester} for
  $O(1/\epsilon^2)$ times implies the result.
\end{proof}

\section{Parameterized and Nonadaptive Algorithms for LIS Estimation}\label{sec:nonadaptive-LIS-algorithms}
Our final goal in this paper is to present sublinear-time approximation
algorithms that, for an array of
length $n$ containing at most $r$ distinct values, 
approximates the length of the LIS within a bounded 
additive error (Theorem~\ref{thm:O(r)}) 
or a bounded multiplicative error (Theorem~\ref{thm:O(root-r)}). 


\subsection{Structural Lemmas and a $\tilde{O}(r)$-Query Nonadaptive Algorithm}\label{sec:O(n)}

We first prove some simple structural 
claims about the restrictions of an LIS to subarrays in an array. 
Consider an array $A$ of length $n$ containing at most $r$ distinct values.
We use $\mathcal{L}$ to denote the set of {\em indices} in an 
arbitrary {\em fixed} LIS of $A$. We refer to the array values
at the indices of $\mathcal{L}$ as the values of $\mathcal{L}$.
Let $t < n$ be a parameter, we partition $[n]$ into $t$
contiguous intervals of equal length $[n]  = \bigcup_{i=1}^t I_i$. This
defines the corresponding $t$ subarrays 
 $A_i, ~ i=1, \ldots ,t$ where $A_i$ is the subarray restricted to $I_i, ~ i=1, \ldots ,t$.
Given a parameter $t' \geq 1$, a subarray $A_i$ is said to be $t'$-\emph{nice}
if the number of distinct values taken by the
indices in $\mathcal{L} \cap A_i$ is 
at most $t'$. 

\begin{claim}\label{cl:15}
  At least $|\mathcal{L}| - nr/(tt')$ indices in $\mathcal{L}$ belong
  to $t'$-nice subarrays of $A$.
\end{claim}
\begin{proof}
 $\mathcal{L}$ has at most one value (for the reader, note; value -- not indexes) in common in distinct
 subarrays, hence there are at most $r/t'$ subarrays $A_i, i=1, \ldots, t$
 for which $\mathcal{L} \cap A_i$ takes more than $t'$ values. 
  Thus, the number of indices in $\mathcal{L}$ that are ``lost to" subarrays that are not $t'$-nice is at most $r/t' \cdot (n/t)$. 
\end{proof}

For $\lambda \in (0,1]$, a subarray 
$A_i, i \in [t]$ is $\lambda$-\emph{dense} if $|\mathcal{L} \cap A_i|
\ge \lambda |A_i|$, and it is $\lambda$-{\em sparse} otherwise.  
Let $\mathcal{L}_{t', \lambda}$ denote the
restriction of $\mathcal{L}$ to $\lambda$-dense $t'$-nice subarrays of $A$. 

\begin{claim}
  \label{cl:17}
        At least $|\mathcal{L}| - nr/(tt') - \lambda n$ indices in
        $\mathcal{L}$ belong to $t'$-nice $\lambda$-dense subarrays
        $A_i, ~ i \in [t]$.
\end{claim}
\begin{proof}
  A $t'$-nice subarray that is not $\lambda$-dense
  contributes at most $\lambda n/t$ to the size of
  $\mathcal{L}$, and so in total, all non-dense $t'$-nice subarrays contribute at
  most $\lambda n$ to size of $\mathcal{L}$. 
        The claim follows, using Claim~\ref{cl:15}.
\end{proof}

\subsubsection{$\tilde{O}(r)$-Query Nonadaptive Algorithm}

The main idea is to use large enough $t$, so that most subarrays
$A_i, ~ i \in [t]$ are $1$-nice and $\lambda$-dense for an appropriate value of
$\lambda$. Then finding the LIS in each subarray is easy, and leads to a good
approximation of the LIS size of $A$.

The formal description of our algorithm is given as Algorithm~\ref{alg:O(n)}. The
algorithm gets oracle access to an array $A$ of length $n$, and gets as inputs, the error parameter $\epsilon \in (0,1)$, and
an upper bound $r$ on
the number of distinct values in $A$.
Informally, we divide the array into $t =4r/\epsilon$ subarrays.
This will make most dense subarrays $1$-nice with respect to a fixed
(unknown) LIS $\mathcal{L}$ (for an appropriate density parameter).
We then sample $O(\log r)$ indices in each
subarray to find the values that are `typical' in each subarray.

Next, our goal is to output as an estimate for $|\mathcal{L}|$, the size of $\mathcal{L}'$ which is the restriction of
$\mathcal{L}$ to such typical values. This will naturally be an
underestimate, but with a small additive error. To estimate the
size of $\mathcal{L}'$, we consider all possible increasing sequences
of the typical values, taking one value from each subarray. Since most
subarrays are $1$-nice, there exists such a sequence of typical values such that the size of an LIS restricted to that sequence is quite close to $|\mathcal{L}'|$. Finally, for a given
$1$-nice subarray $A_i$, the largest subsequence in $A_i$ that takes
one given value $v$ can be easily determined - this is just the distance
to the constant function of value $v$. The details follow.

For a value $v \in \R$ and a set of indices $S
\subseteq [n]$, we denote by $\mathsf{den}(v,S)$, the
density of $v$ with respect to $S$ -- namely, $\mathsf{den}(v,S) = |\{i \in S|~
  A[i]=v\}|/|S|$. We say that $v$ is $\nu$-typical with respect to $S$ if
$\mathsf{den}(v,S) \geq \nu$.

\begin{algorithm}
  \caption{$\tilde{O}(r/\epsilon^3)$ LIS estimation algorithm}
  \begin{algorithmic}[1]
    \Require parameter $\epsilon > 0$; oracle access to array $A$ of length $n$; upper bound $r$ on the number of distinct values in $A$
    \State $\delta \gets \epsilon/8$, $t \gets 4r/\epsilon$
    \State Divide the array into $t$ equal-length subarrays $A_1, A_2, \dots, A_{t}$.
    \For {$i \in [t]$}
    \State {\bf find typical values in each subarray as follows: } Sample uniformly at
                random, and
                independently, a (multi)set of indices
                $S_i$ in $A_i$, of size $s=\frac{1}{\delta^2} \log (6tr)$. Make
                queries to the array values at indices in $S_i$. \label{stp:r-algo-index-sampling}
    \State\label{step:V_i}   Let $S_i^* \subseteq S_i$ contain all indices
                $j \in S_i$ for which $A[j]$ is $(\frac{\epsilon}{4} - \frac{\delta}{2})$-typical with
                respect to $S_i$. 
     Let $V_i^*$ be the corresponding set of array values. \label{stp:r-algo-Vi-formation}
      \EndFor
  \vspace{0.6cm}              
          \State \Comment{\textsf{For a subset $P \subseteq [t],$ of subarrays, we call a sequence of values
                    $(v_i)_{i \in P}$ a {\em pseudo-solution} if
                    $(1)~v_i \in V_i^*~\text{for}~i\in
                    P,~\text{and}~(2)~v_i \le v_j~\text{for}~i<j$}.}
\vspace{0.6cm}                
  \State  For each $~P \subseteq [t]$ and {each pseudo-solution $(v_i)_{i \in P}, $} \label{stp:pseudo-solution-iteration}
    \State\label{step:output} ~ ~ Compute $\sum_{i \in P} \mathsf{den}({v_i,
                  S_i})$. 
    \State Output $\hat{\mathcal{L}}$ to be the maximum
                over all pseudo-solutions of the value obtained in Step \ref{step:output}.
  \end{algorithmic}
  \label{alg:O(n)}
\end{algorithm}

\begin{lemma}
With probability at least $2/3$, the estimate 
output by
Algorithm~\ref{alg:O(n)} is
within $\pm \epsilon n$ of the true length of the LIS in the array
$A$. 
Further, the algorithm is nonadaptive, making
$\tilde{O}(r/\epsilon^3)$ queries.
\end{lemma}
\begin{proof} 
Let $t, s$ and $\delta$ be as defined in
Algorithm~\ref{alg:O(n)}.
Let $E$ denote the event that 
for all $i \in [t]$ and $v \in \R,$ $|\mathsf{den}(v,S_i) -
  \mathsf{den}(v,A_i)| \leq
  \delta/2$.
\begin{claim}\label{clm:sampling-only-heavy-values}
$\Pr(E) \geq 5/6$.
\end{claim}
\begin{proof}
  Consider an arbitrary subarray $A_i$ and a value $v \in
  \mathbb{R}$. The expected value of $\mathsf{den}(v,S_i)$ is equal to
  $\mathsf{den}(v,A_i)$ and hence, by a Hoeffding bound, with probability at most $2\exp(-\delta^2s/12)= \frac{1}{6tr}$, we have $|\mathsf{den}(v,S_i) -
  \mathsf{den}(v,A_i)| \geq
  \delta/2$. Using the union bound over the $t$ distinct
  subarrays and possibly all $r$ values, the difference above is
  bounded for each value in each subarray.
      \end{proof}
Let $V_i^*$ be as defined in Step~\ref{step:V_i} of the algorithm. 
We note that the event $E$ implies that $V_i^*$
contains every $\frac{\epsilon}{4}$-typical value with respect to
$A_i$, and contains no value that is not $(\frac{\epsilon}{4} -
          \delta)$-typical with respect to $A_i$.
 Henceforth, we condition on $E$.
By this conditioning, and the setting of $\delta$, for each $i \in [t],$ we have $|V_i^*|\leq \frac{1}{(\epsilon/4) - \delta}
\leq 8/\epsilon$. 
Hence, the total number of pseudo-solutions 
is at most $\left(\frac{8}{\epsilon}\right)^t$ . 
The following claim is immediate. 
\begin{claim}\label{clm:lis-sufficiently-sampled}
Conditioned on $E$, for every $\epsilon/4$-dense $1$-nice 
subarray $A_i$, the value at the 
indices $\mathcal{L} \cap A_i$ is included in $V_i^*$. \qed
\end{claim}

Let $P \subseteq [t]$ and a corresponding pseudo-solution
$\mathcal{V}= (v_i)_{i \in P}$ (see
Algorithm~\ref{alg:O(n)}). 
An extension of $\mathcal{V}$ is the set 
$\mathcal{I}_{\mathcal{V}} = \bigcup_{i \in P}\{j \in A_i: ~ A[j] = v_i\}$.
Note that for every pseudo-solution 
$\mathcal{V}$, its extension $\mathcal{I}_{\mathcal{V}}$ is an increasing sequence in $A$.
We denote by $\mathcal{L}_{1,\epsilon/4}$, the restriction of the fixed LIS $\mathcal{L}$ to $1$-nice
$\epsilon/4$-dense subarrays.
Claim~\ref{clm:lis-sufficiently-sampled} implies that one of the
pseudo-solutions that Algorithm~\ref{alg:O(n)} considers in
Step~\ref{stp:pseudo-solution-iteration} is such that its
 extension corresponds to $\mathcal{L}_{1,\epsilon/4}$.
Claim~\ref{cl:17} and the setting of $t$ in
Algorithm~\ref{alg:O(n)}
imply that for such a restriction,  
$|\mathcal{L}_{1,\epsilon/4}| \ge |\mathcal{L}| - (\epsilon n/2)$.

\begin{claim}\label{clm:correct-estimation-of-length}
Conditioned on $E$, with probability at least $5/6$, the estimate $\hat{\mathcal{L}}$ output by Algorithm~\ref{alg:O(n)} is within $\pm \epsilon n$ of the length of LIS in the array.
\end{claim}
\begin{proof}
  Consider an arbitrary pseudo-solution $\mathcal{V}= (v_i)_{i
          \in P}$ for some $P \subseteq [t]$, and its corresponding extension.
Conditioned on $E$, it is clear that the cardinality of the extension $\mathcal{I}_{\mathcal{V}}$ of $\mathcal{V}$ is estimated within $\pm \epsilon n/16$ of its true cardinality.
In particular, since Algorithm~\ref{alg:O(n)} considers a pseudo-solution whose extension corresponds to $\mathcal{L}_{1,\epsilon/4}$, 
we have that $\hat{\mathcal{L}} \ge |\mathcal{L}_{1,\epsilon/4}| - \epsilon n / 16$,
which is at least $|\mathcal{L}| - \epsilon n$. 
Moreover, since the extension of every pseudo-solution considered by the algorithm
is an increasing sequence, we also have that $\hat{\mathcal{L}} \le |\mathcal{L}| + \epsilon n /16 \le |\mathcal{L}| + \epsilon n $.
The claim follows.
\end{proof}
\noindent The statement of the lemma follows by applying a union bound to Claim~\ref{clm:sampling-only-heavy-values} and Claim~\ref{clm:correct-estimation-of-length}.
Finally, it is clear from the description that the query 
complexity of our algorithm is $\tilde{O}(r/\epsilon^3)$, and no query
depends on the values of previously made queries. 
 \end{proof}
\begin{remark}
All the above guarantees hold even if Algorithm~\ref{alg:O(n)} were to first sample a set $S$ of $2t\frac{\log(6tr)}{\delta^2}$ indices uniformly and independently at random from the whole array and let $S_i$ be equal to $S \cap A_i$.
The additional step in the analysis is to show that $|S_i|$ is at least $\frac{\log (6tr)}{\delta^2}$, for each $i \in [t]$, which follows from an application of the Chernoff bound. 
In other words, the algorithm is not only nonadaptive, but also \emph{sample-based}, where a sample-based algorithm whose queries are to points that are sampled uniformly and independently at random.
\end{remark}
\begin{remark}
Algorithm~\ref{alg:O(n)} can also be made to estimate
  the length of the LIS on subarrays obtained by either restricting
  the set of indices, or set of values, or both, of the original array. This will be
  used in the next section.
\end{remark}
\begin{remark}
The running time of the algorithm can be made to $\tilde{O}(r^2/\epsilon^3)$ by using dynamic programming in Steps~\ref{stp:pseudo-solution-iteration} and~\ref{step:output}.
\end{remark}

\subsection{$\tilde{O}(\sqrt{r})$-Query Nonadaptive Algorithm}\label{sec:sqrt}
Let $\mathcal{L}$ denote the set of points in an arbitrary and fixed LIS
in the input array $A$. For simplicity of the
presentation, we assume that our algorithm knows a
lower bound $\lambda n$ on $|\mathcal{L}|$. Disregarding this assumption, the algorithm
will output (w.h.p) a lower bound  estimate of the size of an increasing sequence in
$A$. If $\lambda n$ is indeed a bound as assumed, it will be guaranteed
that the estimate is within the multiplicative error that is
stated. Hence $\lambda$ can be checked by running the
algorithm, in parallel, for a geometrically decreasing sequence of
$\lambda$'s. The reader should think of $\lambda < 1$ as a small
constant (although the algorithm works for $\lambda = o(1)$ as well). 

Throughout this section, we visualize the array values as points in an
$r \times n$ grid $G_n$. 
The vertical axis of $G_n$ represents the range $R$ of the array and is labeled with the at most $r$
distinct array values in increasing order and the horizontal axis is labeled with the
indices in $[n]$. 
We refer to an index-value pair in the grid as a point.
The grid has $n$ points, to which we do not have direct access.

We divide the $r\times n$ grid $G_n$ into $y^*$ rows and $x$ columns that 
partitions $G_n$ into a $y^*  \times x$ grid $G'$ of boxes, where $y^* =
\Theta (\sqrt{r})$ and $x = \Theta (\sqrt{r})$ (both depend also 
on $\lambda$).
Specifically, we divide the interval $[n]$ into $x$ contiguous subarrays.
For $i \in [x]$, let $D_i$ denote the $i$-th subarray.
Additionally, we divide the range $R$ into $y^*$ contiguous intervals of array values,
where for $j \in [y^*]$, we use $I_j$ to denote the $j$-th interval when the intervals are 
sorted in the nondecreasing order of values.
The set of boxes in $G'$ is then $\{(I_j,D_i): ~
i \in [x], j \in [y^*]\}$.

The $y^* \times x$ grid of boxes $G'$ induces a poset $\langle \mathcal{P}, \preceq \rangle$ on the $y^* x$
boxes, which is similar to the natural poset defined on $G_n$. 
Namely, for two boxes
in $G'$ (or for two points in $G_n$), we have $(I_j,D_i)
\preceq (D_s,I_t)$ (or $(i,j)
\leq (s,t)$) if $i \leq s$ and $j \leq t$.
The points in $\mathcal{L}$ form a chain in the
above poset in $G_n$.  Further, each chain in the poset
in $G_n$ forms an increasing subsequence in the array $A$. The boxes in
$G'$ through which $\mathcal{L}$ passes also forms a chain in the poset in
$G'$. 

Every chain of boxes in the poset in $G'$ induces a number of chains
in the poset in $G_n$, but of possibly quite different lengths. 

\begin{definition}[Density of a box]\label{def:density}
Let $I \subseteq R$ be a subset of the range $R$ of values and $B$ be a subarray of $A$. The density of the box $(I,B)$, denoted by $\den(I,B)$, is defined to be
the fraction of indices in the subarray $B$ whose values belong to the interval $I$. 
\end{definition}
\noindent In other words, for each box $(I_j,D_i) \in
G',$ its density $\den(I_j,D_i)$ is the
fraction of points in the subarray $D_i$ that land in the box $(I_j,D_i)$.
For $\beta < 1$ (that the reader can think of as a small constant), a box
$(I_j, D_i)$ is said to be $\beta$-dense, if $\den(I_j,D_i) \geq \beta$.

\subsubsection{Forming the grid $G'$ of boxes}

Our goal in this section is to describe a procedure that determines
the grid $G'$ of boxes. Specifically, as we do not
know the range $R$ and only know an upper bound $r$ on its size
$|R|$, we start by forming an approximation of $R$ and an
approximation of the densities of subinterval ranges in $R$ in the
array $A$. 
To do this, we first partition $R$ into $\tilde{O}(\sqrt{r})$ sub-ranges called \emph{layers}.  
For the sake of generality, we describe the procedure for a subarray $B$ of the 
array $A$.

More generally, given a subarray $B$, and a parameter $y$, our goal is to partition the range $R$ into roughly $y$ intervals of roughly
  equal densities, where the densities are with respect to $B$. 
  We note that although
  the size $r$ of the range $R$ might be relatively large, it is
  possible that some values appear in $B$ much more frequently
  than
  others. One of our goals is to identify such values and well-approximate their densities. 
  We now define a `nice' partition as follows. Given $y$ and $B$, a
  nice $y$-partition of the values in $B$ is a partition $R = \cup_{i=1}^{y^*} I_i$, if
   for
  each $i \in [y^*],$ either $I_i$ contains only one value $v_i$ and
  $\den(I_i,B) \geq  \frac{1}{2y}$, or $\den(I_i,B) \leq \frac{2}{y}$.
  In the former case,
  we call $I_i$ a single-valued layer. In the latter,
  we say that $I_i$ is
  a multi-valued layer (although in an extreme case it might
  contain only one value).  
  We also require that $y^*
  \leq 2y$.

  Next, we describe our procedure $\textsc{Layering}(B,y,t)$ that forms a
  $y$-nice partition of a subarray $B$, along with a good approximation of the densities of the
  single-valued layers. 
  This is quite technical, although standard. We
  advise the reader to avoid it on first reading, and assume that, when needed, we have a nice
  partition along with a good approximation to the densities of layers.

  \begin{algorithm}[h]
\begin{algorithmic}[1]
\Procedure{Layering}{$B,y,t$}
\State \textsf{Goal: To divide the set of array values in the subarray $B$ into
  roughly $y$ contiguous intervals of roughly equal densities. The parameter $t$ is
  used to control the success probability.}
\State Sample a set of $\ell = t\cdot y \log y$ indices
$S$ from $B$, uniformly at random and independently. 
\Comment{{\sf Note that a value $v$
is expected to appear  in proportion to its density in the
array. Hence the collection of values obtained is a multiset of size $\ell$.}}

\State We sort the multiset of values $V = \{B[p]: ~ p \in S \}$ to form a strictly increasing sequence $\mathsf{seq}
= (v'_1 <
 \ldots < v'_q)$, where with each $i \in [q]$ we associate a weight $w_i$
 that equals the multiplicity of $v_i'$ in the multiset $V$ of
 values. \Comment{{\sf Note that
 $\sum_{i \in [q]} w_i = \ell$.}}
\State We now partition the sequence $W = (w_1, \ldots ,w_q)$ into maximal
disjoint contiguous subsequences $W_1, \ldots W_{y^*}$ such that for each $j \in [y^*]$, either $\sum_{w \in
  W_j} w <  2t \log y$, or $W_j$ contains only one member $w$ for
  which $w > t \log y$. 

  Note that this can be done greedily as follows. 
If $w_1 > t\log y$ then $W_1$ will contain
only $w_1$, otherwise $W_1$ will contain the maximal subsequence
$(w_1, \ldots, w_i)$ whose
sum is at most $2t \log y$. We then delete the members of
$W_1$ from $W$ and repeat the process. 
For $i \in [y^*]$, let $w(W_i)$ denote the total weight in $W_i$.

Correspondingly, we obtain a
partition of the sequence $\seq$ of sampled values into at most $y^*$ subsequences $\{S_i\}_{i \in [y^*]}$. 
Some subsequences
contain only one value of weight at least $t \log y$ and are called \emph{single-valued}.  
The remaining subsequences are called \emph{multi-valued}.

For a subsequence $S_i$, let $\alpha_i = \min(S_i)$ and $\beta_i = \max(S_i)$.
Let $\beta_0 = -\infty$.
Note that $\alpha_i \leq \beta_i$  and $\beta_{i-1} < \alpha_i$ for all $i \in [y^*]$.

\State For $i \in [y^*]$, we associate with the subsequence $S_i$, an interval (layer) $I_i \subseteq
R$, where $I_i = (\beta_{i-1}, \beta_i]$, and an approximate density
$\widetilde{\den}(I_i,B) = w(W_i)/\ell$. 
\EndProcedure
\end{algorithmic}
\label{proc:layering}
\end{algorithm}

Let $\{I_i\}_{i = 1}^{y^*}$ be the
set of layers that are
created by a call to \textsc{Layering}$(B,y,t)$. 
Recall that $w(W_i) \ge t \log y$ if $I_i$ is a single-valued layer and $w(W_i) < 2t\log y$ if $I_i$ is multi-valued, where $W_i$ denotes the sum of multiplicities of the values in the sample $S_i$.

For a multi-valued layer  $I_i$,  let $E_i$ denote the event that
$\den(I_i,B) < \frac{4}{y}$. 
For a single-valued layer $I_i$ such that $\den(I_i, B) \ge \frac{1}{2y}$, 
let $E_i$ denote the event  that $\frac{\den(I_i, B)}{2} \leq \widetilde{\den}(I_i, B) \leq \frac{3}{2}\den(I_i,B)$. 
For a single-valued layer $I_i$ such that $\den(I_i, B) < \frac{1}{2y}$, let $E_i$ be the event that
$\widetilde{\den}(I_i,B) \leq \frac{3}{2}\den(I_i,B)$. 
Let $E = \bigcap_{i = 1}^{y^*} E_i$. 
The following claim asserts that the layering above well-represents
the structure of the range w.r.t.\ $B$.

\begin{claim}\label{clm:layering}
 \textsc{Layering}$(B,y,t)$ returns a collection of intervals
 $\{I_i\}_{i=1}^{y^*}$ such that,
 $y^* \leq 
    2y$, and $\Pr(E) = 1 - \exp(\Omega(-t))$.
\end{claim}
\begin{proof}
The total number of single-valued and  multi-valued layers of weight at least $t
\log y$     is at most $y$ since each contributed at least a weight of
$t \log y$ to the total weight of $t y \log y$.
  Other layers are multi-valued of
weight less than $t \log y$ and for each such layer $I_i$, it must
be the case that $I_{i-1}$ and $I_{i+1}$ are of weight at least $t\log
y$. It follows that altogether there are at most $y^* \leq 2y$
layers.

Finally, by a standard multiplicative Chernoff bound, one can see that for
an arbitrary $i \in [y^*]$, we have $\Pr[\bar{E_i}] \le \exp(-\Omega(t
\ln y)) = \exp(-t)/y$. Hence, a union bound over $i \in [y^*]$ implies the lower bound on
$\Pr(E)$.
\end{proof}
\ignore{
For an interval $I_i$ that is tagged as a multi-valued layer (less than $2/y$ fraction of the sample falls into it):
  \begin{align*}
  \Pr[\den(I_i,B) \ge \frac{4}{y}] = \Pr[w(W_i) < (1 - \frac{1}{2})\ell \cdot \den(I_i, B)] \leq \exp(-\Omega(\ell \den(I_i, B))) = \exp(-t\log y).
  \end{align*}
}

We now define the grid $G'$ of boxes as follows.
We first use the procedure \textsc{Layering} on the original
array, $B=A$, with parameters $y =
\frac{\sqrt{r}}{\epsilon}$ and $t=O(1)$, where
the value of $t$ is set to ensure a success probability of $99/100$ in Claim~\ref{clm:layering}. This defines the set of
$y^*$ layers that partitions $R$ as $R = \bigcup_{i \in y^*} I_i$.  
Next we partition $[n]$ into $x = \epsilon\cdot \sqrt{r}$ contiguous
intervals $D_1, \ldots D_x$ each of size $n/x$, which defines $G'$ as
the grid of boxes  $\{(I_j, D_i):~ (j,i)\in [y^*] \times [x]\}$ in the $r \times n$ grid, some of which may be
empty, while some may contain many points.

We
set $\beta = \epsilon^3 \lambda$.
Next, our goal is to find all the  $\beta$-dense boxes in $G'$ by
making $\tilde{O}(\sqrt{r})$ queries and then restrict our attention
only to  these boxes. As described in the introduction
to Section \ref{sec:sqrt},  this will not make the
LIS in this restricted array too short. This is made formal in the
following claim.


\begin{claim}
The number of points in $\mathcal{L}$ that belong to boxes that are not $\beta$-dense is at most $\beta n \cdot (1 + 2y/x)$.
\end{claim}
\begin{proof}
The LIS $\mathcal{L}$ passes through at most $x+y^* \leq x+2y$ boxes, and each
box that is not $\beta$-dense contributes at most $\frac{\beta n}{x}$
points to $\cL$.
We conclude that  the boxes that are not $\beta$-dense contribute at most $(x+2y)
\frac{\beta n}{x}  \leq  \beta n \cdot (1 + 2y/x)$ points to the LIS.
\end{proof}

We do not know which boxes are $\beta$-dense. We approximate
this by sampling, and this is formally presented below as algorithm
\textsc{Gridding}. The algorithm assumes the partition of $[n]$ and of the range $R$
as above.
As before, $t = O(1)$ in this procedure can be set appropriately to ensure a large constant success probability.

\begin{algorithm}
\begin{algorithmic}[1]
\Procedure{Gridding}{$A,\{I_j\}_{j \in [y^*]}, \{D_i\}_{i \in [x]}, \beta$}
\For {$i \in [x]$}
\State Sample $\ell = t\cdot\frac{1}{\beta}\cdot\log(\frac{x}{\beta})$ indices from $D_i$ uniformly and independently at random.
\For {$j \in [y^*]$}
\State Label box $(I_j, D_i)$ as \emph{dense} if and only if the values
on at least $\frac{3}{4}\beta \ell$ points from the sample fall into
the box; namely, if for at least $\frac{3}{4}\beta \ell$ indices sampled from $D_i$,
the values are in $I_j$.
\EndFor
\EndFor
\EndProcedure
\end{algorithmic}
\end{algorithm} 


Let $D$ 
be the event that all $\beta$-dense boxes are tagged as \emph{dense} by the
procedure, and that every box that is not $\beta/8$-dense is
not tagged as dense.
\begin{claim}\label{clm:dense-box}
  $\Pr[D] \ge 1 - \exp(-\Omega(t))$.
\end{claim}
\begin{proof}
 Fix $i \in [x]$. The number of $\beta$-dense boxes of the form
 $(\cdot, D_i)$ is at most $1/\beta$. 
 A direct application of the Chernoff bound and union bound over the $\beta$-dense boxes shows that, with probability
 at least $1-\exp(-t)/x$,
 every $\beta$-dense box is correctly tagged as dense by the 
 algorithm.
 Using the union bound over $i \in [x]$ ends the proof.
The fact that no box with density smaller than $\beta/8$ is tagged as dense also follows by a simple Chernoff bound, followed by a union bound, where the union bound is over the at most $4x/3\beta$ boxes that are tagged dense by the procedure.
\end{proof}

From now on, we assume that we have the grid $G'$ for which the events
$E$ and $D$ hold. This is the initialization of our data structure as described
in the introduction to Section~\ref{sec:sqrt}.
We now refine the data structure as follows. 

A $\beta$-dense box may have density that is anything in
$[\beta,1]$. For a better approximation guarantee, we need to identify the regions with density nearly equal to $\beta$. 
To achieve this, we perform the following finer layering using the
procedure \textsc{Layering} on each dense box. 

\medskip

\noindent
{\bf Finer layering of each dense box:}
For each $i \in [x]$, call \textsc{Layering} on the array $D_i$ with
$y = 1/\beta$ and $t' = \Theta(\log (x/\beta))$. In this case,  we do not collapse
the single-valued intervals into a single layer, but rather just leave them as different layers of
the same value and density $\beta$. 

Let $\{I'_{k}\}_{k \in [y^*_i]}$ be the set of intervals returned by the procedure. 
We restrict our attention to the boxes $(I'_{k},D_i)$ that are contained in some $\beta$-dense box $(I_j, D_i)$,
and refer to them as $\beta$-dense cells. 

Fix $D_i$.
The number of $\beta$-dense cells in $D_i$ is at most $2y = 2/\beta$.
Claim~\ref{clm:layering} asserts that with probability $1-
\exp(-\Omega(t'))= 1- \frac{\beta}{100 x}$ each $D_i$ is layered so
that each $\beta$-dense cell has true density at most $3\beta/2$ and at least $\beta/8$. 
Additionally, the portion of a $\beta$-dense box that is not 
covered by $\beta$-dense cells has true density smaller than $\beta$. 
This implies that for all $i \in
[x]$ this happens with probability at least
$99/100$. We denote this
event by $F$, and assume in what follows that $F$ happens.

\subsubsection{Chain reduction}\label{sec:chain-reduction}

In this section, we define a poset over dense cells and argue that in
order to well-approximate the LIS, it is enough to restrict our attention to LIS's in a few chains in this poset.

Since dense cells, by definition, are contained inside dense boxes, we
denote dense cells using triplets $(I_j, D_i, k)$, where this triplet denotes the $k$-th dense cell inside the dense box $(I_j, D_i), i \in [\epsilon\sqrt{r}], j \in [\sqrt{r}/\epsilon]$.

Recall that there is a poset $\langle \mathcal{P}, \preceq\rangle$ on the dense boxes. Now, we
define another poset $\langle \mathcal{P}^\star, \preceq^\star\rangle$ whose elements are the (at most) $2x/\beta$ dense cells.
The order relation $\preceq^\star$ is defined by  $(I_j, D_i, k)
\preceq^\star (I_{j'}, D_{i'}, k')$ if and only if either $(I_j,D_i) \neq
(I_{j'}, D_{i'})$ and $(I_j,D_i) \preceq
(I_{j'}, D_{i'})$, or if $j'=j, i'=i$ and $k \leq k'$.
Note that the poset $\preceq^\star$ is not consistent
with a grid poset, it rather inherits the order from $\mathcal{P}$ for
cells in different boxes.

Let $\mathcal{L}_1$ be the LIS $\mathcal{L}$ restricted to dense
boxes, let $C(\mathcal{L}_1,\mathcal{P})$ be the set of dense boxes in which
$\mathcal{L}_1$ passes, and let $C(\mathcal{L}_1,\mathcal{P}^\star)$ be the set of dense cells in which
$\mathcal{L}_1$ passes.  We observe that $C(\mathcal{L}_1, \mathcal{P})$ and
$C(\mathcal{L}_1,\mathcal{P}^\star)$ are chains in the corresponding
posets.

Our goal now is to show that there are a small number of chains in
$\mathcal{P}^\star$ that cover
$C(\mathcal{L}_1,\mathcal{P}^\star)$. This is done as follows.

\begin{itemize}
\item For parameter $\tau = 5/\lambda$, remove repeatedly
  antichains of size larger than $\tau$ from $\mathcal{P}^\star$. Here, by
  removing, we mean the deletion of the  points in the
  cells of the corresponding antichain from the array\footnote{For a
     single-valued cell $a=(I_j,D_i,k)$ taking the value $v$, there might be other points with value $v$ in
    the dense box containing $a$. When we remove $a$ from
    $\mathcal{P}^\star$, we `mentally' remove some $\beta n/x$ points of
    value $v$ from the box $(I_j,D_i)$. This is not algorithmically done, but
    will just be used in the analysis.}. 
\end{itemize}
Let the resulting poset be denoted by $\mathcal{P}^{\star
  \star}$. The maximum antichain in this poset has size at most $\tau$, and
 Dilworth's theorem implies that there is a decomposition of
$\mathcal{P}^{\star \star}$ into at most $\tau$ chains. 
These chains, being made of dense cells, is naturally extended to
at most $\tau$  (possibly intersecting) chains of the poset
$\mathcal{P}$. Let these chains be $C_1, \ldots
,C_\tau$. 
We bound the `loss' to the LIS incurred by chain reduction in Claim~\ref{clm:chain-reduction}.
\begin{claim}\label{clm:chain-reduction}
Conditioned on the events $E \cup D \cup F$, the number of
points in $\cL_1$ that does not belong to $\cup_{i = 1}^{\tau} C_i$  is at most $4n/\tau$.
\end{claim}
\begin{proof}
  There are at most $2x/\beta$ dense cells and  hence at most
$\frac{2x}{\tau \beta}$ antichains can be removed. Since $C(\mathcal{L}_1,
\mathcal{P}^\star)$ is a chain, it loses at most one dense cell each
time an antichain is removed, and altogether $\frac{2x}{\tau \beta}$
dense cells.
But each dense cell has at most $\frac{3\beta n}{2x}$ points in it, which
implies that  the loss to $\mathcal{L}_1$ is at most $\frac{3n}{\tau}$.
\end{proof}

Let $\mathcal{L}_2$ denote the LIS $\mathcal{L}_1$ after chain reduction. The following claim is straightforward.
\begin{claim}\label{cl:est1}
There exists $i \in [\tau]$ such that $|\text{LIS}(C_i)| \ge \frac{|\mathcal{L}_2|}{\tau}$.
\end{claim} 

We point out that no query is made at this stage.

\subsubsection{Estimating the LIS restricted to poset chains}\label{sec:chains2}
Let $\mathcal{P}'$ denote the poset obtained from $\mathcal{P}$ after
removing the large antichains in $\mathcal{P}^\star$.
At this point, we have covered the poset $\mathcal{P}'$ using at
most $\tau$ chains $C_1, \ldots ,C_\tau$.  This reduces the LIS
estimation to estimating the LIS in one of these chains.

In what follows, we fix
such one chain $C$, and denote by $\mathcal{L}(C)$ the LIS in the
array restricted to $C$.
The chain in $C$ is
        composed of a sequence of horizontal and vertical blocks,
        arranged in a staircase manner (see Figure \ref{fig:chain1}), 
        where a horizontal block is a sequence of contiguous boxes in the chain
        from the same layer, and a vertical block is a sequence of 
        contiguous dense boxes in the chain that belong to the same subarray.
          \begin{figure}
    \begin{center}
    \includegraphics[scale=0.3]{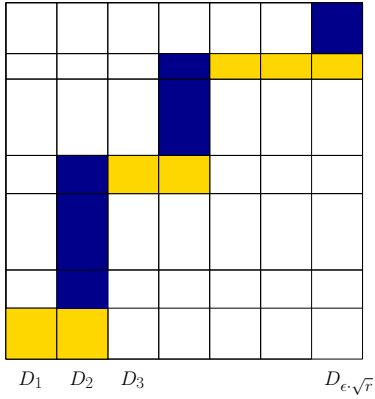}
  \end{center}
    \caption{A staircase like chain, and its decomposition
                  into two chains, one that contains only horizontal
                  blocks and one that contains only vertical
                  blocks. In each of the chains, no two blocks share a
                  layer or a subarray.}
    \label{fig:chain1}
    \end{figure} 

 Let $H_1,H_2$ be two maximal horizontal blocks in $C$. Blocks $H_1,H_2$ have a subarray
 $D_i$ in
 common if there are boxes $(I_j,D_i) \in H_1$ and $(I_s,D_i) \in
 H_2$. In particular, there is a vertical block between them. Two horizontal chains have at most one
 subarray in common, and if this happens, then the common subarray
 $D_i$ defines the rightmost box of the `lower' horizontal chain (the horizontal chain in the
 lower layer) and the leftmost box of the `upper' horizontal chain.
 We conclude that if we arrange the horizontal blocks from bottom to
 top as $H_1, \ldots H_s$, and remove the rightmost box from $H_i$ if it
 has a common subarray with $H_{i+1}$, we get a sequence of horizontal
 blocks in which no two share a subarray. We use $C_H$ to denote this subchain of
 $C$. Notice that $C_V= C \setminus C_H$  is a
 chain that contains only vertical blocks, where no two share a
 layer (see Figure \ref{fig:chain1}, as how the whole chains $C$
 decomposes into a chain of horizontal blocks and a chain of vertical blocks).

To estimate the size of $\mathcal{L}(C)$, we estimate the LIS within
$C_H$ and $C_V$ separately, and use the larger for the size estimate
for $\mathcal{L}(C)$.
In the following, we denote  $\mathcal{L}_H$ and $\mathcal{L}_V$ for  LIS$(C_H)$
and LIS$(C_V)$, respectively.
 The main advantage of this decomposition of $C$ into $C_H$ and $C_V$ is given by the following
 observation.
 \begin{observation}
   For any chain $C,$ we have $|\mathcal{L}_H|= \sum_{B \in C_H} |\text{LIS}(B)|$, where the
   sum is over the horizontal blocks $B$ in $C_H$. A similar statement holds
   for $C_V$ in which case, horizontal blocks are replaced with
   vertical blocks.
 \end{observation}

The observation leads to an immediate adaptive way to approximate the
lengths of $\mathcal{L}_H$
and $\mathcal{L}_V$. 
We just sample a constant number of
blocks from the chain, estimate the LIS in each block, and normalize to
estimate the LIS in the entire chain. By the Hoeffding bound, we can see 
that the estimate is accurate enough with high probability. The adaptivity is need
to locate each block, and to estimate the LIS within a block
(horizontal and vertical), which we did not yet specify how to do. To avoid
adaptivity, we will rely on the fact that if $\mathcal{L}_V$ is large,
then it must contain a large number of small vertical blocks. Thus
sampling uniformly in $[n]$ will hit  such blocks frequently enough to
facilitate the Hoeffding bound above. Further, for the estimation of LIS within 
a short vertical block, we will use the algorithm from Section
\ref{sec:O(n)}. For horizontal blocks we need some further relaxations.
We will show below that we may restrict ourselves to short horizontal
blocks, due to the choice of parameters in the formation of the grid $G'$.
 This again will facilitate the use of the algorithm from
Section \ref{sec:O(n)}. The details now follow.

\medskip
\noindent{\bf Estimating the length of LIS in a horizontal chain.}
A horizontal block belonging to a multi-valued layer is referred to as
 a multi-valued horizontal block, and a horizontal block belonging to a
 single-valued layer is a single-valued horizontal block.
 We treat these horizontal blocks separately.

Let  $m=\epsilon/\lambda^2$. 
Let $\mathcal{L}'_H$ denote the restriction of $\mathcal{L}_H$ after
deleting multi-valued horizontal blocks containing more than $m$ boxes.
We first show that the length of $\cL'_H$ is not much smaller than $\cL_H$.

\begin{claim}
$|\mathcal{L}'_H| \le |\mathcal{L}_H| - 4\lambda^2 \epsilon n$.
\end{claim}  
\begin{proof}
Recall that, conditioning on event $E$, the density of a
multi-valued layer is at most $4/y = 4\epsilon/\sqrt{r}$. The number 
of multi-valued horizontal blocks that each contain at least $m$ 
boxes is at most  $x/m =
\lambda^2\sqrt{r}$.
Hence the total number of points from $\mathcal{L}_H$ that belong to
horizontal blocks containing at least $m$ boxes and belonging to multi-valued layers, is 
 at most $\frac{4\epsilon n}{\sqrt{r}}  \cdot
\lambda^2\sqrt{r} \leq 4\lambda^2 \epsilon n$.
\end{proof}

Let $C'_H$ denote the chain obtained after removing multi-valued horizontal blocks containing more than $m$ boxes.
Next, we argue that either there are many multi-valued horizontal blocks in
the chain $C'_H$, or we can remove all multi-valued horizontal blocks without
losing too much in the LIS. 
Indeed, 
assume that there are at most $\phi = \epsilon \lambda^2 y$ multi-valued horizontal blocks in the chain $C'_H$.
Then, by removing all of them, we end up losing only $\phi \cdot
\frac{4n}{y} \leq 4n\epsilon \lambda^2$ points (as each
multi-valued layer has density at most $4/y$). 

If we are in the case that we have at least $\phi$ multi-valued horizontal blocks
in $C'_H$, then the average number of values in such a horizontal block
is at most $\frac{r}{\phi} \leq \frac{r}{\epsilon \lambda^2 y}  =
\frac{\sqrt{r}}{\lambda^2}$ by our choice of $y$. 
That is, with probability at least $1 - \frac{1}{100\log (\tau)}$, a uniformly random multi-valued horizontal
block in $C'_H$ contains at most $\frac{100\sqrt{r} \log
  (\tau)}{\lambda^2}$ values. 
Thus, we have reduced the problem to estimating the LIS in a collection of (possibly very long) single-valued horizontal blocks and several short multi-valued horizontal blocks containing $O(\frac{\sqrt{r}\log (\tau)}{\lambda^2})$ values. 

In the following, we use the term \emph{segment} to denote a subarray composed of $2m$ subarrays $\{D_i, D_{i+1}, \dots, D_{i+2m-1}\}$ for some $i\in [x-2m+1]$. 
A segment is said to contain a multi-valued horizontal block $H$ if all the subarrays forming $H$ are contained in the 
segment.  

Fix  $r'=\frac{100\sqrt{r} \log (\tau)}{\lambda^2}$. 
Our algorithm for estimating the length of $\cL_H$ is as follows:

\begin{enumerate}
\item Sample $t\log^2 (\tau)$ uniformly random segments.
\item For each sampled segment $B$, query $s = \Theta\left(\frac{m \log (\tau)}{\beta}\cdot \frac{r'}{\epsilon^3 \lambda^6} \log \left(\frac{r'^2}{\epsilon \lambda}\right)\right)$ points uniformly and independently at indexes from $B$ and run Algorithm~\ref{alg:O(n)} with parameters $r'$ (for the number of distinct values) and $\epsilon \lambda^2$ (for approximation guarantee) using the samples that fall into the multi-valued horizontal block $H$ contained in the segment $B$, if any.
\item Estimate the contribution to the LIS from multi-valued horizontal blocks by summing the answers returned by the algorithm in the previous steps and then normalizing appropriately.
\item Estimate the contribution to the LIS from single-valued horizontal blocks by summing the estimates of the densities of all single-valued horizontal blocks in $C_H$ (as we already know these estimates from the \textsc{Gridding} stage).
\item Output an estimate $L_H$ of the length of $\cL_H$ by summing the above two estimates. 
\end{enumerate}

Clearly, the contribution to $\cL_H$ from single-valued horizontal blocks is estimated within  multiplicative $(1 \pm \frac{1}{2})$-error, by our conditioning on the event $F$.
We show the following.
\begin{claim}\label{clm:hor-chain}
With probability $1 - O(\frac{\log (\tau)}{\tau^2})$, the contribution to $\cL_H$ from multi-valued horizontal blocks is estimated within an additive error of $\epsilon \lambda^2 n$.
\end{claim}
\begin{proof}
As we already argued, the length of an LIS in $C_H'$ is at least $|\cL_H| - 8\lambda^2 \epsilon n$.

There are at least $\epsilon \lambda^2 y (1 - \frac{1}{100\log (\tau)})$ multi-valued horizontal blocks consisting of at most $r'$ values in $C_H'$.
Each one of at least half among these horizontal blocks are contained in at least $m$ segments in the array.
As there are at most $x$ many segments in the array, the probability that a uniformly random segment does not contain any multi-valued horizontal block consisting of at most $r'$ values from $C_H'$ is at most $1 - \frac{99m\epsilon \lambda^2 y}{200 x} \leq 1 - \frac{99}{200}.$
Hence, the probability that no multi-valued horizontal segment is contained in any of $\frac{400\log (\tau)}{99}$ independently sampled uniformly random segments is at most $\frac{1}{\tau^2}$.
By a simple union bound, with probability at least $1 - \frac{\log (\tau)}{\tau^2}$, the procedure samples at least $\log (\tau)$ segments that each contain a multi-valued horizontal block consisting of at most $r'$ values from $C'_H$. 
In the rest, we condition on sampling at least $\log (\tau)$ segments that each contain a multi-valued horizontal block consisting of at most $r'$ values from $C'_H$.

For a specific such segment, the fraction of points in that segment that belong to the contained multi-valued horizontal block is at least $\frac{\beta}{16m}$. 
This is because a multi-valued horizontal block contains at least one box tagged as dense by the \textsc{Gridding} procedure and by our conditioning (event $D$), each box tagged as dense has density at least $\beta/8$.
Thus, for a specific segment $B$, the probability that fewer than $s\beta/16$ points from the sample in Step 2 land into the corresponding multi-valued horizontal block is, by a Chernoff bound, at most $\exp(-s\beta/64) \le 1/\tau^2$.
By a union bound, the probability that there exists a segment such that the multi-valued horizontal block contained inside it 
has fewer than $s\beta/16$ points from the sample is at most $t \log^2(\tau)/\tau^2$. In the rest of the proof, we condition on the event that \emph{sufficiently many} sampled points fall into the multi-valued horizontal blocks corresponding to the sampled segments.

Now, the probability that an invocation of Algorithm~\ref{alg:O(n)} succeeds, when run with a sample of size $s\beta/16$,
is at least $1 - \frac{1}{\tau^2}$. 
The probability that all invocations of Algorithm~\ref{alg:O(n)} succeed is, by the union bound, at least $ 1 - \frac{\log^2(\tau)}{\tau^2}$.

Finally, by a direct application of the Hoeffding bound, we can bound the error due to the sampling of blocks to be an additive $\epsilon \lambda^2 n$ with probability at least $1 - \frac{1}{\tau^2}$.

By a union bound over all of the above ``bad" events above, we ultimately get that, the estimate output by our procedure is within $\pm 2 \lambda \epsilon n$ of the contribution of multi-valued horizontal blocks to $|\cL_H|$, with probability at least $1 - \frac{3\log (\tau)}{\tau^2}$.
\end{proof}

  \medskip

\noindent{\bf Estimating the length of the LIS in a vertical chain.}
  Let $\nu = \epsilon \lambda^2$.
  We may assume that the vertical chain is composed of at least $\nu \cdot x$ vertical blocks, for otherwise, we can abandon the entire vertical chain by incurring a `loss' to the LIS amounting to at most $\nu \cdot n$ points.
  Additionally, since the boxes from different vertical blocks belong to different layers, using a similar averaging argument as before, we can show that with probability at least $1 - \frac{1}{100\log(\tau)}$, a uniformly random vertical block
  contains at most $\frac{100\sqrt{r}\log (\tau)}{\lambda^2}$ distinct values. 

  Therefore, in order to estimate the length of the LIS in the vertical chain, we sample $O(\log (\tau))$ 
  subarrays $D_i, i \in [x]$ and run the pseudo-solution-based LIS estimation algorithm, restricted to the 
  vertical box, if any, that belongs to this subarray
  while making sure that the success probability is at least $1 - \frac{1}{100\log \tau}$ and the error parameter is $\epsilon\lambda^2$.
  The details of how to implement this procedure nonadaptively are identical to how we implemented the estimation of the LIS in $C_H$ in the preceding section. The query complexity is also identical.

\begin{claim}\label{clm:ver-chain}
 With probability $1 - O(\frac{\log (\tau)}{\tau^2})$, we estimate the contribution of vertical blocks to within an additive error of $\epsilon \lambda^2 n$.
\end{claim}
The proof of the above claim is analogous to that of Claim~\ref{clm:ver-chain} and is therefore, omitted.
        \begin{remark}
          We have analyzed this section for any arbitrary bound on the
          number of distinct values in the array $r$.  
 We note that for the general case when $r=n$, the
          algorithm becomes simpler. The case of estimating the LIS in
          a (possibly long) single-valued layer does not change as this
          was done by a simple computation using the already available
          data. For the case of short horizontal blocks, or vertical
          blocks, they are spanned by constantly many subarrays (or just one, in the
          case of vertical block). For these we had to run the
          algorithm from Section \ref{sec:O(n)}. But when $r= n$ then
          the subarrays have size $n/x = O(\sqrt{n})$ and hence we can
          afford to query {\em all} indices in these constantly many
          subarrays, and compute the LIS in the corresponding blocks exactly.
        \end{remark}
  \subsubsection{Correctness, approximation guarantee, and query complexity}
  In this section, we complete the analysis of our algorithm and finish the proof of Theorem~\ref{thm:O(root-r)}.

\medskip
  \noindent{\bf Success probability.} The probability that any of Layering, Gridding and Finer Gridding fail is at most $3/100$. 
  For a specific chain of boxes, by Claims~\ref{clm:hor-chain} and~\ref{clm:ver-chain}, we know that estimating the length of LIS within them is within the approximation guarantee with probability at least $1 - O(\frac{\log(\tau)}{\tau^2})$.
  By a union bound over all $\tau$ chains, we can see that the probability of incorrectly estimating the LIS length in some chain is at most $1/100$. Thus, overall, the failure probability is at most a small constant.

\medskip
  \noindent{\bf Query complexity.} The query complexity is clearly
        $\tilde{O}(\sqrt{r}\cdot \text{poly}(1/\lambda))$  from the description of the algorithm.

\medskip
  \noindent{\bf Approximation guarantee.} 
  Consider a fixed true LIS $\mathcal{L}$.
  The loss due to ignoring boxes that are not $\beta$-dense ($\beta = \epsilon^2 \lambda$) is at most $\epsilon^3 \lambda n + \epsilon \lambda n$.
  The loss due to antichain removal is at most $4n/\tau$, which is equal to $4\lambda n/5$.
  The resulting increasing sequence has length at least $|\mathcal{L}| - \epsilon^3 \lambda n - \epsilon \lambda n - 4\lambda n / 5$, which is at least $(1 - \epsilon^3 - \epsilon - 4/5) \cdot |\mathcal{L}|$, by our assumption on $\lambda$.

  After chain decomposition, the length of the LIS in the best chain is at least $(1 - \epsilon^3 - \epsilon - 4/5) \cdot |\mathcal{L}|/\tau$, which is equal to $\frac{\lambda}{5}\cdot |\mathcal{L}| \cdot (1/5 - \epsilon^3 - \epsilon)$.
  Since we split the chains into horizontal and vertical chains, we further lose a factor of $2$, and the resulting LIS length becomes $\frac{\lambda}{10}\cdot |\mathcal{L}| \cdot (1/5 - \epsilon^3 - \epsilon)$.

  In case of horizontal chains, we additionally lose a $8\epsilon \lambda^2 n$ and in the case of vertical chains, we additionally lose $\epsilon \lambda^2 n$. 

  That is the length of LIS in the (best) horizontal chain is at least $\frac{\lambda}{10}\cdot |\mathcal{L}| \cdot (1/5 - \epsilon^3 - 9\epsilon)$.
  Finally, using Claims~\ref{clm:hor-chain} and~\ref{clm:ver-chain}, we can see that we estimate the lengths of the best horizontal and vertical chains to within a constant multiplicative factor.
  Overall, the approximation guarantee is multiplicative $\Omega(\lambda)$.

\bibliographystyle{alpha}
\bibliography{references}
\end{document}